\documentclass[12pt, reqno]{amsart}
\usepackage{latexsym, enumerate, imakeidx,mathrsfs, color,amscd,amsmath,amsfonts,amssymb,geometry,color,amsthm, float}
\usepackage[english]{minitoc}
\usepackage[bookmarksnumbered, colorlinks, plainpages]{hyperref}
\hypersetup{colorlinks=true,linkcolor=black, anchorcolor=green, citecolor=cyan, urlcolor=red, filecolor=magenta, pdftoolbar=true}

\usepackage[T1]{fontenc}

\newtheorem{theorem}{Theorem}[section]
\newtheorem{lemma}[theorem]{Lemma}
\newtheorem{proposition}[theorem]{Proposition}
\newtheorem{corollary}[theorem]{Corollary}
\theoremstyle{definition}
\newtheorem{definition}[theorem]{Definition}

\theoremstyle{remark}
\newtheorem{remark}[theorem]{Remark}
\numberwithin{equation}{section}

\allowdisplaybreaks 

\DeclareMathOperator{\rang}{ran\,\!}

\usepackage{physics}

\makeatletter
\def\blfootnote{\xdef\@thefnmark{}\@footnotetext}

\makeatother

\begin{document}
\title[Nussbaum-Szko{\l}a Distributions in Semifinite von Neumann Algebras]{Quantum \lowercase{\texorpdfstring{\ensuremath{f}}{}}-divergences via  Nussbaum-Szko{\l}a Distributions in Semifinite von Neumann Algebras}

\author[T. Anastasiadis]{Theodoros Anastasiadis}\blfootnote{The article is part Anastasiadis' PhD dissertation, prepared at the University of South Carolina under the supervision of Androulakis.}
\address{University of South Carolina\\
Department of Mathematics\\
1523 Greene St., Columbia  SC 29208 USA}
\email{anastast@email.sc.edu}

\author[G. Androulakis]{George Androulakis}
\address{University of South Carolina\\
Department of Mathematics\\
1523 Greene St., Columbia \\
 SC 29208 USA}
\email{giorgis@math.sc.edu}


 


\begin{abstract}
In this article, we prove that the quantum \texorpdfstring{\ensuremath{f}}{}-divergence between two normal states on a semifinite von~Neumann algebra is equal to the classical \texorpdfstring{\ensuremath{f}}{}-divergence between two corresponding classical states, which are called Nussbaum-Szko{\l}a distributions. This result has been proved by the second named author and T.C.~John for normal states on the von~Neumann algebra $\mathbb{B}(\mathscr{H})$ of all bounded operators on a Hilbert space $\mathscr{H}$. We extend their result for normal states on any semifinite von~Neumann algebra, not only $\mathbb{B}(\mathscr{H})$. 
\end{abstract}

\maketitle
 \section{Introduction}

The main result of the present article is to extend the concept of Nussbaum-Szko{\l}a distributions to  semifinite von~Neumann algebras and to prove that the \texorpdfstring{\ensuremath{f}}{}-divergence between any two normal states on such algebras is equal to the  classical \texorpdfstring{\ensuremath{f}}{}-divergence between the 
corresponding Nussbaum-Szko{\l}a distributions. These
distributions were originally introduced in~\cite{Nussbaum-Szkola-2009}, in order to study error exponents in quantum hypothesis testing problems. In~\cite {Nussbaum-Szkola-2009}, the result was initially proved for the von~Neumann algebra $\mathbb{B}(\mathscr{H})$ of bounded operators on a finite-dimensional Hilbert space $\mathscr{H}$ and for the Umegaki relative entropy. Since then, these distributions have been used many times in the literature. A few instances in which they have been useful are the following: multiple state discrimination \cite{Nussbaum2011-uo,audenaert-mosonyi-2014,Li2016-tn}, Gaussian state discrimination \cite{Mosonyi-2009}, study of \texorpdfstring{\ensuremath{f}}{}-divergences between states on finite-dimensional $C^{*}$-algebras \cite{Hiai2011-qb,Hiai2017-sg}, second order asymptotics for quantum hypothesis testing \cite{Li2014-xg,Tomamichel-Hayashi-2013,Datta2013-yg}, investigation of the large derivation regime \cite{Dalai2013-yl}, coherence distillation \cite{Hayashi2021-kr}, derivation of quantum inequalities from their classical analogues \cite{Dupuis2019-vu,Androulakis2023Jul} and examination of finiteness of the Petz-R\'{e}nyi $\alpha$-relative entropy ~\cite{androulakis-john-2024, androulakis-john-2024-lmp}.
The fact that Nussbaum-Szko{\l}a distributions can be used to compute divergences was extended in \cite{Androulakis2023Jul} for any Hilbert space (finite or infinite-dimensional) and for the general \texorpdfstring{\ensuremath{f}}{}-divergence. In this article, we extend the result for normal states on any semifinite von~Neumann algebra, not only $\mathbb{B}(\mathscr{H})$. Some examples of semifinite von~Neumann algebras are $\mathbb{B}(\mathscr{H})$ for any Hilbert space $\mathscr{H}$, the space $L_{\infty} (X,\mu)$ of essentially bounded measurable functions on a measure space $(X,\Sigma,\mu)$ where $\mu$ is a $\sigma$-finite measure, and any type $II$ factor.

In order to state our main result  we introduce some non-technical concepts in this section. We begin with the definition of classical states on a measurable space.

\begin{definition}\label{defn:Classical states} Let $(X,\Sigma)$ be a measurable space, $\nu$ be a measure on $(X,\Sigma)$ and $p:X\to[0,+\infty)$ be a measurable function such that   
$$   
\int_{X} p d\nu=1. 
$$
The probability measure $P$ that has density $p$ with respect to the measure $\nu$, i.e. 
$$
P(Y)=\int_{Y}pd\nu \quad \text{for all }Y\in \Sigma,
$$
is called a {\bf{classical state}} on $(X,\Sigma)$. If there is a countable set $Y \in \Sigma$ such that 
$P\left(X\backslash Y\right)=0 $, then $P$ is a \textbf{discrete} classical state.  
\end{definition} 
Next, we introduce the definition of quantum states.

\begin{definition}\label{defn:quantum states} If $\mathscr{A}$ is a unital $C^*$-algebra, then a linear functional $\phi:\mathscr{A}\to \mathbb{C} $ is called a {\bf{quantum state}}, or simply a {\bf{state}}, on $\mathscr{A}$ if and only if $\phi$ is {\bf{positive}} (i.e. $\phi (a)\geq 0$ for all $a\geq 0$ in $\mathscr{A}$) and $\phi(1)=1$.  
\end{definition}

In our work, we will focus on normal states on a semifinite von~Neumann algebra (see Definition~\ref{defn:normal states}). Note that every von~Neumann algebra is a unital  $C^*$-algebra and that  normal states are special cases of quantum states.

Divergences are tools which are used to discriminate two classical or normal quantum states. Examples of divergences include
the Kullback-Leibler divergence \cite{kullback-leibler-1951}, the Umegaki relative entropy \cite{umegaki-1962}, 
the R\'{e}nyi \cite{renyi1961} and  Petz-R\'{e}nyi $\alpha$-relative entropy \cite{petz1986quasientropy}, the Hellinger divergence, 
the $\chi^2$ divergence and the total variation (see \cite{Liese-Vajda-2006, Csiszar-Sheilds-2004} for the last three divergences). 
The notion of \texorpdfstring{\ensuremath{f}}{}-divergence for a convex  function $f$ generalizes all the above special cases of divergences. 
A rich exposition of  \texorpdfstring{\ensuremath{f}}{}-divergences between classical states can be viewed in \cite{Liese-Vajda-2006, Csiszar-Sheilds-2004}.  
The \texorpdfstring{\ensuremath{f}}{}-divergence between two normal states on a von~Neumann algebra  was introduced by Petz in~\cite{petz1986quasientropy, Petz1985}, by extending Araki's definition for the relative entropy (see \cite{araki1976relative,araki1977relative}). 

The main result of our article is the following theorem.

\begin{theorem}[Main Theorem] \label{maintheorem}
Let $\phi$, $\omega$ be two normal states on a semifinite von~Neumann algebra $\mathscr{M}$ in standard form. Then, there exist a $\sigma$-finite measure space $(X,\Sigma,\nu)$ and functions $f_\phi$, $f_\omega:X\to [0,+\infty)$ (where both depend on both states $\phi$ and $\omega$) which are densities (with respect to $\nu$) of some classical states $P_\phi$ and $P_\omega$, respectively, satisfying \[S_{f}\left(\phi\|\omega\right)=D_{f}\left(P_\phi\|P_\omega\right)\] for every convex function $f:(0,+\infty)\to \mathbb{R} $ where
$S_{f}(\phi\|\omega)$, $D_{f}\left(P_\phi\|P_\omega\right)$ respectively denote the quantum  \texorpdfstring{\ensuremath{f}}{}-divergence between $\phi$ and $\omega$ and the classical  \texorpdfstring{\ensuremath{f}}{}-divergence between $P_\phi$ and $P_\omega$. We call the classical states $P_\phi$ and $P_\omega$  {\bf{Nussbaum-Szko{\l}a distributions}} associated with the states $\phi$ and $\omega$.
\end{theorem}
The proof of Theorem~\ref{maintheorem} can be found  towards the end of Section 4 (after the proof of Proposition~\ref{prop: f phi,omega are distributions}).

This result is a  far reaching generalization of \cite[Theorem~3.8]{Androulakis2023Jul} where only states on the von~Neumann algebra $\mathbb{B}(\mathscr{H})$ for a Hilbert space $\mathscr{H}$ were 
considered. States on $\mathbb{B}(\mathscr{H})$ correspond to trace-class operators on $\mathscr{H}$ which are compact and have necessarily countable spectrum. 
This simplifies the study of the relative modular operator. In addition, this causes the Nussbaum-Szko{\l}a distributions on $\mathbb{B}(\mathscr{H})$ to be discrete. The contribution of our article is that we present  Nussbaum-Szko{\l}a distributions for the first time in  semifinite von~Neumann algebras. 

An essential tool for defining $f$-divergences between normal states is the relative modular operator, with respect to these states, introduced by Araki in~\cite{araki1976relative} and ~\cite{araki1977relative}. In the case of a general von~Neumann algebra, there is no practical formula for the relative modular operator and this is an important obstacle when dealing with the quantum \texorpdfstring{\ensuremath{f}}{}-divergence. However, {\L}uczak, Pods{\k{e}}dkowska and Wieczorek in \cite{Luczak2022Sep} discovered such a formula for the relative modular operator in the case of  semifinite von~Neumann algebras. This is the reason why in the present article we prove our results for normal states on a semifinite von~Neumann algebra and not a general von~Neumann algebra.

The structure of our article is as follows. In Section 2 we provide a review of general prerequisites on von~Neumann algebras and spectral theory which include definitions about traces, measurable operators, non-commutative $L_p$-spaces, the standard form of von~Neumann algebras, the Spectral Theorem and the Borel functional calculus. In Section 3 we initially  review some recent results about the relative modular operator and then use the multiplicative form of the Spectral Theorem for strongly commuting normal operators, to create a bridge between the non-commutative and the commutative $L_2$-spaces. We use that bridge in Theorem~\ref{prop: application of the multiplication form of spectral theorem} where we prove a key formula  for the relative modular operator with respect to normal states on a semifinite von~Neumann algebra.  In Section 4 we provide  the definition of the classical and quantum $f$-divergences and prove a series of results that lead to the proof of our main theorem towards the end of the section. In Section 5 we provide a few applications of our result. In Section 6 we discuss some closing remarks.

\section{Prerequisites on von~Neumann algebras and spectral theory}
In this section, we review some basic definitions and theorems that will be used in the remainder of the article.
We begin  with the definition of a trace on a von~Neumann algebra, which is essential in the definition of  semifinite von~Neumann algebras. 

\begin{definition}\label{defn:Traces} Let $\mathscr{M}^{+}$ be the set of positive operators in the   von~Neumann algebra $\mathscr{M}$. A {\bf{trace}} on $\mathscr{M}$ is a function $\tau: \mathscr{M}^{+}\to [0,+\infty]$ such that \begin{enumerate}
    \item $\tau(ta)=t\tau(a)$ for all $a\in \mathscr{M}^{+}$ and $t\geq0$ where $0(+\infty):=0$,
    \item $\tau(a+b)=\tau(a)+\tau(b)$ for all $ a$, $b\in \mathscr{M}^{+}$,
    \item $\tau(a^{*}a)=\tau(aa^{*})$ for all $a\in \mathscr{M}$.
\end{enumerate}
The trace $\tau$ is called
\begin{enumerate}
\item {\bf{semifinite}} if the set $\mathfrak{N}_{\tau}:=\{a\in \mathscr{M}: \tau(a^{*}a)<+\infty\}$ is $\sigma$-weakly dense in $\mathscr{M}$, 
\item {\bf{faithful}} if $\tau(a^{*}a)=0\Rightarrow a=0$ for all $ a\in \mathscr{M}$, 
\item {\bf{normal}} if $\tau(\sup_{i} \{a_i\})=\sup_{i}\{\tau(a_i)\}$ for every bounded increasing net $(a_i)_i\subseteq \mathscr{M}^{+}$.
\end{enumerate}
\end{definition} 

\begin{definition}\label{defn:Semifinite von Neumann algebra} A von~Neumann algebra $\mathscr{M}$ is called {\bf{semifinite}} if  it possesses a faithful semifinite normal trace. 
\end{definition}

The non-commutative $L_p$-spaces play an important role in our analysis. We review them in Definition~\ref{defn:noncommutative Lp}. Before that, we need two more definitions. 

\begin{definition}\label{defn:affiliated operators} Let $\mathscr{K}$ be a Hilbert space, $\mathscr{M}\subseteq \mathbb{B}(\mathscr{K})$ be a von~Neumann algebra and $D(T)$ be the domain of a linear operator $T:D(T)\subseteq \mathscr{K}\to \mathscr{K}$. We say that $T$ is {\bf{affiliated}} with $\mathscr{M}$ if $UTU^{-1}=T$ for every unitary $U\in \mathscr{M}'$, where  $\mathscr{M}':=\{x\in \mathbb{B}(\mathscr{K}):xm=mx,\text{ for all } m\in \mathscr{M}\}$ is the commutant of $\mathscr{M}$. The set of all operators affiliated with $\mathscr{M}$ is denoted by $\mathscr{M}_\eta$.
\end{definition}
 Affiliated operators are essential for introducing  measurable operators in the following definition. 

\begin{definition}\label{defn:Measurable operators} Let $\mathscr{K}$ be a Hilbert space and $\mathscr{M}\subseteq \mathbb{B}(\mathscr{K})$  be a semifinite von~Neumann algebra with a faithful semifinite normal trace $\tau$. By $\mathscr{M}_{proj}$, we denote the set of orthogonal projections that belong to $\mathscr{M}$, i.e. $\mathscr{M}_{proj}:= \{p\in \mathscr{M}: p=p^2=p^{*}\}$. For $\epsilon$, $\delta>0$ we define the set 
\[N(\epsilon,\delta):= \{T\in \mathscr{M}_{\eta}: \exists \text{ } p\in \mathscr{M}_{proj} \text{ such that } \rang  (p) \subseteq D(T), \text{ }\|Tp\|\leq \epsilon, \text{ }\tau(1-p)\leq \delta \}. \]
A closed densely-defined operator $T$ affiliated with $\mathscr{M}$ is called {\boldmath{$\tau$}}-{\bf{measurable}},  if for every $\delta>0$ there exists $\epsilon>0$ such that $T\in N(\epsilon,\delta)$. We denote the set of $\tau$-measurable operators by $\widetilde{\mathscr{M}}$ (see \cite{Nelson1974,Hiai2021Apr}).
\end{definition}
According to \cite[Proposition 4.10]{Hiai2021Apr}, $\widetilde{\mathscr{M}}$ is a $*$-algebra with respect to the adjoint $^*$, the strong sum $\overline{T_1+T_2}$ and the strong product $\overline{T_1T_2}$. For every $T_1$, $T_2\in \widetilde{\mathscr{M}}$, we will use the common convention that $T_1+T_2$ and $T_1T_2$ denote the strong sum $\overline{T_1+T_2}$ and the strong product $\overline{T_1T_2}$ respectively.

\begin{definition}\label{defn:noncommutative Lp} Let $\mathscr{M}$ be a semifinite von~Neumann algebra with a faithful semifinite normal  trace $\tau$ and $p\in (0,+\infty]$. The {\bf{non-commutative}} {\boldmath{$L_p$}}{\bf{-space}} associated with $(\mathscr{M},\tau)$ is denoted by $L_p(\mathscr{M},\tau)$ and is defined as 
\begin{align}\label{eq: definition of Lp} L_{p}(\mathscr{M},\tau):&=\{a\in \widetilde{\mathscr{M}}: \tau(|a|^p)<+\infty\} \quad \text{ if}\quad  0<p<+\infty \text{ and } \\
 L_{\infty}(\mathscr{M},\tau):&=\mathscr{M}.
\end{align}
\end{definition}
For $p\in [1,+\infty]$, $L_{p}(\mathscr{M},\tau)$ is a Banach space with respect to the norm \[\|a\|_p:=\begin{cases}
\tau(|a|^p)^{\frac{1}{p}} & \text{if } 1\leq p<+\infty\\
\quad \|a\| & \text{if }  p=+\infty
\end{cases}
.\] If $p=2$, $L_{2}(\mathscr{M},\tau)$ is a Hilbert space with inner product defined by the formula 
\begin{align}\label{eq: inner product of L2}
\braket{a}{b}_{L_2(\mathscr{M},\tau)} =\tau(a^{*}b)   \text{   for all } a,\text{ } b\in L_{2}(\mathscr{M},\tau).\end{align}
To simplify the notation, we occasionally denote the inner product of $L_2(\mathscr{M},\tau)$ by $\braket{\cdot}{\cdot}$ instead of $\braket{\cdot}{\cdot}_{L_2(\mathscr{M},\tau)}$.
The notation $\braket{\cdot}{\cdot}_{L_2(\mathscr{M},\tau)}$ is kept when different inner products appear nearby.

 \begin{remark}
[Generalized H{\"o}lder's inequality  for the non-commutative $L_p$-spaces]\cite[Proposition 4.43]{Hiai2021Apr} Let $\mathscr{M}$ be a semifinite von~Neumann algebra with a faithful semifinite normal  trace $\tau$ and let $p$, $q$, $r\in (0,+\infty]$   with $\frac{1}{p}+\frac{1}{q}=\frac{1}{r}$.  Let $a\in L_p(\mathscr{M},\tau)$ and $b\in L_q(\mathscr{M},\tau)$. Then, $ab\in L_r(\mathscr{M},\tau)$ and \[\|ab\|_{r}\leq \|a\|_{p}\|b\|_{q}.\]
  For example, if $a$, $b\in L_2(\mathscr{M},\tau)$ then $ab\in L_1(\mathscr{M},\tau)$, and if $c\in L_p(\mathscr{M},\tau)$ and $m\in \mathscr{M}$, then   $mc\in L_p(\mathscr{M},\tau)$ for any $p\in (0,+\infty]$.
\end{remark}
\begin{remark} Observe that Equation~\eqref{eq: inner product of L2}   requires an extension of the trace $\tau$ to elements of $L_1(\mathscr{M},\tau)$. This extension is possible and moreover, 
the map $\tau: L_1(\mathscr{M}, \tau) \to \mathbb{C}$ is a positive linear functional  
(see \cite[Proposition 4.37]{Hiai2021Apr}) i.e.
\begin{equation*}
    \begin{gathered}
\tau (\lambda a + b) = \lambda \tau (a) + \tau (b) \text{ for all } 
\lambda \in \mathbb{C} \text{ and } a,\ b \in L_1(\mathscr{M}, \tau), \text{ and}\\
\tau(c) \geq 0 \text{ for all positive self-adjoint operators } c \text{ in } L_1(\mathscr{M}, \tau),
\end{gathered}
\end{equation*}
and satisfies the cyclic property (see \cite[Proposition 4.45]{Hiai2021Apr}) i.e.,
\begin{equation} \label{Eq:cyclicitity}
\tau (ab) = \tau (ba)
\end{equation}
 for all $a \in L_p(\mathscr{M}, \tau)$  and $b \in L_q(\mathscr{M},\tau)$  where $ p,\ q \in [1, \infty]$  with $\frac1p +\frac1q =1$.
For the details on how this extension can be done as well as a comprehensive introduction to measurable operators and non-commutative $L_p$-spaces, we refer the reader to \cite{Nelson1974, Hiai2021Apr,Yeadon1975Jan,TakesakiII,Terp1981}. 
\end{remark}

Next, we review normal states and the predual of a von~Neumann algebra. The reason why we will restrict ourselves to normal states is that the quantum \texorpdfstring{\ensuremath{f}}{}-divergence is defined  between two normal positive functionals.  

\begin{definition}\label{defn:normal states} Let $\mathscr{M}$ be a von~Neumann algebra. A positive linear functional $\phi$ on $\mathscr{M}$   is {\bf{normal}} if $\phi(\sup_{i} \{a_i\})=\sup_{i}\{\phi(a_i)\}$ for every bounded increasing net $(a_i)_i\subseteq \mathscr{M}$. In particular, $\phi$ is called a {\bf{normal state}} on $\mathscr{M}$ if it is a normal positive linear functional with $\phi(1)=1$. 
\end{definition}
For any von~Neumann algebra $\mathscr{M}$ the {\bf{predual}} $\mathscr{M}_*$ is the set of all linear functionals on $\mathscr{M}$ that are $\sigma$-weakly continuous. A positive linear functional $\phi$ on $\mathscr{M}$ belongs to $\mathscr{M}_{*}$ if and only if $\phi$ is normal (see \cite[Theorem 3.6.4]{Pedersen2018}). In other words, $\phi$ is a normal state if and only if $\phi\in \mathscr{M}_{*}^+$ and $\phi(1)=1$ where $\mathscr{M}_{*}^{+}:=\{x\in \mathscr{M}_* : x \text{ is }  \text{a } \text{positive } \text{functional}\}$ (the {\bf{positive cone}} of the predual $\mathscr{M}_*$).

We review the standard form of a von~Neumann algebra. Recall that a cone in a Hilbert space $\mathscr{H}$ is a convex set that is invariant under scalar multiplication by non-negative scalars. Recall that a map $J:\mathscr{H}\to \mathscr{H}$ is called:
\begin{itemize}
    \item {\bf{conjugate-linear}} if it is additive and $J(\lambda \xi)=\overline{\lambda}J(\xi)$ for  all $\xi\in \mathscr{H} $  and $ \lambda \in \mathbb{C}$, 
    \item {\bf{isometric}} if $\|J(\xi)\|=\|\xi\|$ for all $\xi\in \mathscr{H}$,
    \item {\bf{involution}} if $J^2=1$.
\end{itemize}
 According to Haagerup \cite[Theorem~1.6]{Haagerup1975}, every von~Neumann algebra is isomorphic to a von~Neumann algebra $\mathscr{M}$ on a Hilbert space $\mathscr{H}$, such that there exists a map $J:\mathscr{H}\to \mathscr{H}$ which is a conjugate-linear isometric involution  and there exists a closed  cone $P$ in $\mathscr{H}$ which is  self-dual (i.e. $P=\{\xi\in \mathscr{H}:\braket{\xi}{\eta}\geq 0 \text{ for all } \eta \in P\}$) such that the following four properties are satisfied:
\begin{enumerate}
    \item $J\mathscr{M}J=\mathscr{M}'$,
    \item $JxJ=x^*$ for all $x\in \mathscr{M}\cap \mathscr{M}'$,
    \item $J(\xi)=\xi$ for all $\xi\in P$ and
    \item $xJxJ(P)\subseteq P$ for all $x\in\mathscr{M}$.
\end{enumerate}
The quadruple $(\mathscr{M},\mathscr{H},J,P)$ is called a standard form of the von~Neumann algebra and plays an important role in defining the relative modular operator, which in turn is needed to define the quantum \texorpdfstring{\ensuremath{f}}{}-divergence. 

In \cite[Theorem 2.3]{Haagerup1975}, Haagerup proves that the standard form of a von~Neumann algebra is unique in the following sense. If $(\mathscr{M}, \mathscr{H}, J, P)$ and $(\tilde{\mathscr{M}}, \tilde{\mathscr{H}}, \tilde{J}, \tilde{P})$ are two standard forms and $\Phi: \mathscr{M}\to \tilde{\mathscr{M}}$ is a $^*$-isomorphism, then there exists a unique unitary $u:\mathscr{H}\to \tilde{\mathscr{H}}$ such that 
\begin{enumerate}
    \item $\Phi(x)=uxu^{-1}$ for all $x\in \mathscr{M}$,
    \item $\tilde{J}=uJu^{-1}$ and
    \item $\tilde{P}=u(P)$.
\end{enumerate}

\begin{remark}\label{standard form in semifinite and left multiplication}
  In \cite[Remark 4.52]{Hiai2021Apr}, \cite[Theorem 36]{Terp1981},  a standard form is presented for any semifinite von~Neumann algebra $\mathscr{M}$ with a faithful semifinite normal trace $\tau$. This standard form is $\left(\pi(\mathscr{M}),L_2(\mathscr{M},\tau),J,L_2(\mathscr{M},\tau)_{+}\right)$, where the von~Neumann algebra $\mathscr{M}$  acts on the Hilbert space $L_2(\mathscr{M},\tau)$ by the left multiplication {$\pi:\mathscr{M}\to \mathbb{B}\left(L_2(\mathscr{M},\tau)\right)$} which is defined as
    \begin{align}\label{Eq: Left Multiplication} \pi(x)(\xi):=x\xi \text{ for all } x\in \mathscr{M} \text{ and } \xi\in L_2(\mathscr{M},\tau).
    \end{align} 
    The conjugate-linear  involution $J: L_2(\mathscr{M},\tau)\to L_2(\mathscr{M},\tau)$ is defined as 
    \begin{align*}
        J(\xi):&=\xi^*  \text{ for all } \xi\in L_2(\mathscr{M},\tau) 
        \end{align*}
         and the closed self-dual cone is defined as
        \begin{align*}
        L_2(\mathscr{M},\tau)_+:&=\{\xi\in L_2(\mathscr{M},\tau): \xi \text{ is a positive self-adjoint operator}\}.
    \end{align*}
  \end{remark}

   \begin{remark}
   Throughout the article, we consider every semifinite von~Neumann algebra $\mathscr{M}$ in the standard form $\left(\pi(\mathscr{M}),L_2(\mathscr{M},\tau),J,L_2(\mathscr{M},\tau)_{+}\right)$. The von~Neumann algebra $\mathscr{M}$ is isometric to $\pi(\mathscr{M})$ (see \cite[Section 6.11]{StratilaZsido2019}) and for the remainder of the article we identify $\mathscr{M}$ with $\pi(\mathscr{M})$, which is a common practice in the literature. 
\end{remark}

\begin{remark}\label{existence of hphi,homega} According to \cite[Theorem 3.12]{Hiai2021Apr} and \cite[Lemma 2.10]{Haagerup1975}, if $(\mathscr{M},\mathscr{H},J,P)$ is a standard form of a von~Neumann algebra, then for every $\omega \in \mathscr{M}_{*}^{+}$ there exists a unique $\xi_\omega\in P$ such that $\omega(x)=\braket{\xi_\omega}{x (\xi_\omega)}_{\mathscr{H}}$,   for all $x\in \mathscr{M}$. The vector  $\xi_\omega$ is called the {\bf{vector representative}} of $\omega$. In particular, if $\mathscr{M}$ is a semifinite von~Neumann algebra with a faithful semifinite normal trace $\tau$, then we have $\xi_\omega \in L_2(\mathscr{M},\tau)_+$. Using Equations~\eqref{eq: inner product of L2}, \eqref{Eq:cyclicitity} and the fact that $\xi_\omega $ is self-adjoint, we obtain
\begin{align}\label{eq: equation for vector representative} \omega(x)=\braket{\xi_\omega}{x\xi_\omega}_{L_2(\mathscr{M},\tau)}=\tau(\xi_\omega x \xi_\omega)=\tau(h_\omega x)=\tau(xh_\omega) \text{ for all } x\in \mathscr{M},
\end{align}
where  $h_\omega:=\xi_\omega \xi_\omega^{*}=\xi_\omega^2 $ is a positive self-adjoint operator that belongs to $L_1(\mathscr{M},\tau)$. In particular, $\tau(h_\omega)=\omega(1)$ and so if $\omega$ is a state we know that $\tau(h_\omega)=1$.
\end{remark}

We continue with a review on the Spectral Theorem  for self-adjoint unbounded operators. 

\begin{theorem}\cite[Theorem 5.7]{Schmudgen} \label{thm: spectral decomposition}
  Let $\mathscr{K}$ be a Hilbert space and let $T:D(T)\subseteq \mathscr{K}\to \mathscr{K}$ be a self-adjoint operator. Then, there exists a unique spectral measure $E_T(\cdot)$ on the Borel $\sigma$-algebra $Borel(\mathbb{R})$ such that $T=\int_{(-\infty,+\infty)} \lambda dE_{T}(\lambda).$ This formula is called the {\bf{spectral decomposition}} of $T$. 
\end{theorem}
  Next, we recall the Borel functional calculus. 
\begin{theorem}\cite[Section 5.3]{Schmudgen} \label{thm: Borel functional calculus}
  Let $\mathscr{K}$ be a Hilbert space, $T:D(T)\subseteq \mathscr{K}\to \mathscr{K}$ a self-adjoint operator and  $T=\int_{(-\infty,+\infty)} \lambda dE_{T}(\lambda)$  the spectral decomposition of $T$. If $f:\mathbb{R}\to \mathbb{R}$ is a Borel measurable function, then
  $f(T):= \int_{(-\infty,+\infty)} f(\lambda) dE_{T}(\lambda)$ is a normal operator with a dense domain
  $D\left(f(T)\right)= \Bigl\{ k\in \mathscr{K}: \int_{(-\infty,+\infty)} |f(\lambda)|^2d\braket{k}{E_{T}(\lambda)k}_{\mathscr{K}}<+\infty \Bigr\}.$
\end{theorem}
Note that if the operator $T$ is also positive self-adjoint, then $(-\infty,+\infty)$ in Theorems~\ref{thm: spectral decomposition} and \ref{thm: Borel functional calculus} can be replaced by $[0,+\infty)$. 
For a further review on the spectral theory, we refer the reader to \cite{Schmudgen}. 

We close this section of prerequisites with two lemmas which  will be used in a few proofs in Sections 3 and 4.

\begin{lemma}\label{lem: applying the Borel functional calculus}
    Let $\mathscr{K}_1$, $\mathscr{K}_2$ be Hilbert spaces,  $T:D(T)\subseteq \mathscr{K}_1\to \mathscr{K}_1$ be a self-adjoint operator and $U:\mathscr{K}_1\to \mathscr{K}_2$ be a unitary operator. If $f:\mathbb{R}\to \mathbb{R}$ is a Borel measurable function, then $Uf(T)U^{-1}=f(UTU^{-1})$ with the domains related by   $
D\!\left(f(UTU^{-1})\right)
=
U\!\left(D\!\left(f(T)\right)\right).
$
\end{lemma}
\begin{proof}
Let $T=\int_{(-\infty,+\infty)} \lambda dE_{T}(\lambda)$ be the spectral decomposition of the operator $T$ where $E_T(\cdot)$ is the spectral measure of $T$. The operator $S:=UTU^{-1}: U\left(D(T)\right)\subseteq \mathscr{K}_2\to \mathscr{K}_2$ is 
self-adjoint and therefore possesses a spectral decomposition with a unique spectral measure $E_{S}(\cdot)$. We observe that the measure $UE_{T}(\cdot)U^{-1}$ is an orthogonal projection-valued measure  and in addition
\[
\int_{(-\infty,+\infty)} \lambda d\left( UE_{T}(\lambda) U^{-1}\right)=U\left( \int_{(-\infty,+\infty)} \lambda d E_{T}(\lambda) \right) U^{-1}=UTU^{-1}=S.
\]
Therefore, by the uniqueness of the spectral measure, $E_S(\cdot)=UE_T(\cdot)U^{-1}$. Next, we prove the equality of the domains.
\begin{align*}
    D\left(f(UTU^{-1})\right)&= \Bigl\{ k_2\in \mathscr{K}_2: \int_{(-\infty,+\infty)} |f(\lambda)|^2d\braket{k_2}{UE_{T}(\lambda)U^{-1}(k_2)}_{\mathscr{K}_2}<+\infty \Bigr\} \\
    &= \Bigl\{ k_2\in \mathscr{K}_2: \int_{(-\infty,+\infty)} |f(\lambda)|^2d\braket{U^{-1}(k_2)}{E_{T}(\lambda)U^{-1}(k_2)}_{\mathscr{K}_1}<+\infty \Bigr\}  \\ &(\text{because } U^{*}=U^{-1} \text{ since } U \text{ is unitary}) \\ &= \Bigl\{ U(k_1)\in \mathscr{K}_2: k_1 \in \mathscr{K}_1 \text{,  } \int_{(-\infty,+\infty)} |f(\lambda)|^2d\braket{k_1}{E_{T}(\lambda)k_1}_{\mathscr{K}_1}<+\infty \Bigr\}\\
    &(\text{by setting } k_2=U(k_1) \text{ since the unitary } U:\mathscr{K}_1\to \mathscr{K}_2 \text{ is 1-1 and onto})\\
    &=U\left(D\left(f(T)\right)\right).
\end{align*}
Now, we prove that $Uf(T)U^{-1}=f(UTU^{-1})$. By Theorem~\ref{thm: Borel functional calculus}, we obtain
\begin{align*}
f(UTU^{-1})
&= \int_{(-\infty,+\infty)} f(\lambda)\, d\bigl(U E_T(\lambda) U^{-1}\bigr)\\&= U\left( \int_{(-\infty,+\infty)} f(\lambda) d E_{T}(\lambda) \right) U^{-1}\\&=Uf(T)U^{-1}.
\end{align*}
\end{proof}

The following lemma is well known, but we present its proof for completeness.

\begin{lemma} \label{lem:E zero infty}
Let $\mathscr{K}$ be a Hilbert space, $T:D(T)\subseteq \mathscr{K}\to \mathscr{K}$ be a positive self-adjoint operator. Let $T=\int_{[0,+\infty)} \lambda dE_{T}(\lambda)$ be the 
spectral decomposition of $T$ according to Theorem~\ref{thm: spectral decomposition} and the comment following Theorem~\ref{thm: Borel functional calculus}. Then, $E_T((0,+\infty))$ is the smallest projection $p$ on $\mathscr{K}$ that satisfies 
$pT=T$.
\end{lemma}

\begin{proof}
Since $E_{T}(\cdot)$ is a projection-valued measure, $E_{T} \left( (0,+\infty)\right)$ is a projection. We first prove that $E_T((0,+\infty))T = T$.  By  $\mathscr{X}_{A}(\cdot)$ we denote the characteristic function of any set $A$ and observe that
\[E_{T}\left( (0,+\infty)\right)=\int_{(0,+\infty)} 1dE_{T}(\lambda)= \int_{[0,+\infty)}\mathscr{X}_{(0,+\infty)} (\lambda)dE_{T}(\lambda) \quad \text{and infer}\] 
\begin{align}\label{eq: lem 2.15 first equation}
E_{T}\left( (0,+\infty)\right) T &=\int_{[0,+\infty)}\mathscr{X}_{(0,+\infty)} (\lambda)dE_{T}(\lambda)\int_{[0,+\infty)} \lambda dE_{T} (\lambda) \notag \\&\subseteq \overline{\int_{[0,+\infty)}\mathscr{X}_{(0,+\infty)} (\lambda)dE_{T}(\lambda)\int_{[0,+\infty)} \lambda dE_{T} (\lambda)} \notag \\
&= \int_{[0,+\infty)}\mathscr{X}_{(0,+\infty)} (\lambda) \lambda dE_{T} (\lambda) \quad \text{(by \cite[Theorem 4.16 (iii)]{Schmudgen})} \notag \\
&=\int_{[0,+\infty)} \lambda dE_{T}(\lambda)=T.
\end{align}
Since  $E_{T}\left( (0,+\infty)\right)$ is bounded, it has a full domain and therefore the domain of the composition $E_{T}\left( (0,+\infty)\right)T$ is equal to the domain of $T$. Thus,  Inclusion~\eqref{eq: lem 2.15 first equation}  is an equality and we have 
\begin{align}\label{eq: lem 2.15 firrrst equation}
    E_{T}\left( (0,+\infty)\right) T=T.
\end{align}

Let $p$ be a projection on $\mathscr{K}$ that satisfies $pT=T$. It suffices to show $\rang \left( E_T((0,+\infty)) \right) \subseteq \rang (p)$.
Since $pT=T$, we infer 
\begin{align} \label{eq:lem 2.15 second equation}
     &\rang (T)\subseteq \rang\left( p \right) \notag\\
    \Rightarrow \quad &\overline{\rang (T)}\subseteq \overline{\rang\left( p \right)}=\rang\left( p \right) \notag \\
    \Leftrightarrow \quad &\text{ker}(T)^{\perp} \subseteq \rang\left( p \right) .
\end{align}
{\bf{\underline{Claim:}}}
\begin{align}\label{eq: lem 2.15 third equation}
    \text{ker}(T) =\rang (E_T(\{ 0 \} )).
\end{align}
Since $\lambda \mathscr{X}_{\{ 0 \} } (\lambda)   =0$ for every $\lambda \geq 0$ and $\mathscr{X}_{\{ 0 \} } (T) =E_T(\{ 0 \}) $,
we obtain by \cite[Theorem 5.9 (6)]{Schmudgen} that $\overline{T E_T(\{ 0 \})} =0$. Since $T$ is closed and $E_T(\{ 0 \} )$ is bounded
we have that $T E_T(\{ 0 \} )$ is closed, hence $T E_T( \{ 0 \} ) =0$. Thus, $\rang (E_T\left(\{ 0 \}\right) ) \subseteq \text{ker} (T)$. 

Now, for showing the reverse inclusion, let $x\in \text{ker}(T)$. Then, 
\begin{align*}
    0=\|T (x)\|_{\mathscr{K}}^2&=\int_{[0,+\infty)} \lambda^2d\braket{x}{E_{T}(\lambda)x}_{\mathscr{K}} \quad \text{(by \cite[Theorem 5.9 (2)]{Schmudgen})}\\
    &=\int_{(0,+\infty)} \lambda^2d\braket{x}{E_{T}(\lambda)x}_{\mathscr{K}}, 
\end{align*}
which implies 
\begin{align*}
  &0=\braket{x}{E_{T}\left( (0,+\infty) \right) (x)}_{\mathscr{K}}= \|E_{T}\left( (0,+\infty) \right) (x)\|_{\mathscr{K}}^2\\
  \Leftrightarrow \quad &E_{T}\left( (0,+\infty) \right) (x)=0
   \Leftrightarrow  E_{T}\left( (0,+\infty) \right) (x)+E_{T}\left( \{0\} \right) (x)=E_{T}\left( \{0\} \right) (x) \\
   \Leftrightarrow \quad &E_{T}\left( [0,+\infty)\right) (x) =E_{T}\left( \{0\} \right) (x)  
   \Leftrightarrow   x= E_{T}\left( \{0\} \right) (x).
\end{align*}
Consequently, $x\in \rang\left(E_{T}\left( \{0\} \right) \right)$ and this finishes the proof of the Claim. 

We apply Equation~\eqref{eq: lem 2.15 third equation} on Inclusion~\eqref{eq:lem 2.15 second equation} and derive
  \begin{align*}
&\rang\left(E_{T}\left( \{0\}\right)\right)^{\perp} \subseteq \rang\left( p \right)\\
     \Leftrightarrow \quad &\rang\left(E_{T}\left( (0,+\infty) \right)\right)\subseteq \rang\left( p \right), 
\end{align*} 
which proves that $E_{T}\left( (0,+\infty) \right)$ is smaller than any projection $p$ on $\mathscr{K}$ that satisfies $pT=T$.
\end{proof}

\section{The Relative modular operator}
In this section, we focus on the relative modular operator with respect to two normal states on a semifinite von~Neumann algebra. This is an essential tool for the quantum \texorpdfstring{\ensuremath{f}}{}-divergence. We begin by reviewing its definition for general von~Neumann algebras. 
 
\begin{definition}\cite[Section 2.1]{Hiai-2018}\label{relative modular operator}
    Let  $(\mathscr{M},\mathscr{H},J,P)$ be a standard form of a von~Neumann algebra. Let $\phi,\text{ }\omega\in  \mathscr{M}_{*}^{+}$ and $\xi_\phi$, $\xi_\omega$ be the vector representatives of $\phi$ and $\omega$ respectively (see Remark~\ref{existence of hphi,homega}). By $s_\mathscr{M}(\omega)$ we denote the orthogonal projection onto 
$\overline{\mathscr{M}'\xi_{\omega}}$ and by 
  $s_{\mathscr{M}'}(\omega) $ the orthogonal projection onto 
$\overline{\mathscr{M}\xi_{\omega}}$. Note that $s_{\mathscr{M}}(\omega)$ is also the support of the state $\omega$ on $\mathscr{M}$, by \cite[Section 5.22]{StratilaZsido2019}, i.e. $s_{\mathscr{M}}(\omega)= 1-\sup\{e\in \mathscr{M}_{proj}:\omega(e)=0\}$ (see~\cite[Section 5.15]{StratilaZsido2019}). Recall that $s_{\mathscr{M}}(\omega)\in \mathscr{M}$ and  $s_{\mathscr{M}'}(\omega) \in \mathscr{M}'$, by~\cite[Section 3.8]{StratilaZsido2019}.
The operator $S_{\phi,\omega}$ is defined as 
\begin{align}\label{eq: S of phi,omega}
 S_{\phi,\omega}(x\xi_{\omega}+\eta):=s_{\mathscr{M}}(\omega)x^{*}\xi_{\phi} \text{ for all } x\in \mathscr{M}, \text{ } \eta\in \left(1-s_{\mathscr{M}'}(\omega)\right)\mathscr{H}. \end{align}
 According to \cite[Lemma 2.2]{araki1977relative}, $ S_{\phi,\omega}$ is a densely-defined closable conjugate-linear operator. The {\bf{relative modular operator with respect to {\boldmath{$\phi$}} and {\boldmath{$\omega$}}}}  is defined as
\begin{align}\Delta_{\phi,\omega}:=S_{\phi,\omega}^{*}\overline{S_{\phi,\omega}} \end{align} where the bar denotes the closure of $S_{\phi,\omega}$.
\end{definition} 
The polar decomposition of $\overline{S_{\phi,\omega}}$ is given by \cite[Theorem 2.4]{araki1977relative}
\begin{equation}\label{eq: polar decomp of S}
\overline{S_{\phi,\omega}}=J\Delta_{\phi,\omega}^{\frac{1}{2}}.
\end{equation}
The relative modular operator $\Delta_{\phi,\omega}$ is a positive self-adjoint operator with support projection $s_\mathscr{M}(\phi)s_{\mathscr{M}'}(\omega)$.

There are specific examples of von~Neumann algebras for which a useful formula for the relative modular operator has been established. 
Hiai in \cite[Example 2.5 and Example 2.6]{Hiai-2018} mentions a formula for the relative modular operator which helps him derive a formula for the quantum \texorpdfstring{\ensuremath{f}}{}-divergence between normal states in the case of an abelian von~Neumann algebra  and in the case of the von~Neumann algebra $\mathbb{B}\left(\mathscr{H}\right)$, respectively. See also \cite[Appendix B]{Androulakis2023Jul}. 

{\L}uczak, Pods{\k{e}}dkowska and Wieczorek in \cite{Luczak2022Sep} analyze the relative modular operator with respect to normal states on a semifinite von~Neumann algebra. Their analysis on this topic is one of the key ingredients in our work, especially their Theorem $8'$. To present this theorem, it is necessary to review some other important tools first (following~\cite{Luczak2022Sep}).

Let $\mathscr{M}$ be a semifinite von~Neumann algebra with a faithful semifinite normal trace $\tau$. The right multiplication $\pi'$ of $\mathscr{M}$ on $L_2(\mathscr{M},\tau)$,     $\pi':\mathscr{M}\to \mathbb{B}\left(L_2(\mathscr{M},\tau)\right)$, is defined by the formula
\begin{align}\label{Eq: right multiplication}
\pi'(x)(\xi):=\xi x \text{ for all } x\in \mathscr{M} \text{ and } \xi\in L_2(\mathscr{M},\tau). 
\end{align}
The left multiplication $\pi$, which was presented in Equation~\eqref{Eq: Left Multiplication}, is also useful. We  extend $\pi,\text{ }\pi'$ to  $\tau$-measurable self-adjoint, possibly unbounded operators $x$ on $ L_{2}(\mathscr{M},\tau)$. Let $x=\int_{(-\infty,+\infty) } t dE(t)$ be the spectral decomposition of such $x$. Then, $\pi(x)$ and $\pi'(x)$ are defined as
\begin{align}\label{eq: pi,pi' of x}
    \pi (x):=\int_{(-\infty,+\infty)} td\pi\left(E(t)\right) \text{ and } \pi' (x):=\int_{(-\infty,+\infty)} td\pi'\left(E(t)\right),
\end{align}
where the spectral measures $\pi\left(E(\cdot)\right),\text{ }\pi'\left(E(\cdot)\right)$ are defined according to Equations~\eqref{Eq: Left Multiplication}, \eqref{Eq: right multiplication} respectively.
 By \cite[Formulae (3), (4)]{Luczak2022Sep}, if $f:\mathbb{R}\to \mathbb{R}$ is any Borel measurable function, then we obtain
     \begin{align}\label{Eq: pi of f of x}
         \pi(f(x))=f(\pi(x))\quad and \quad
         \pi'(f(x))=f(\pi'(x)).
     \end{align}
A useful form of $\pi(x),\text{ } \pi'(x)$ is given in the following proposition.
 \begin{proposition}\label{prop:pi,pi' for self-adjoint}\cite[Proposition 4]{Luczak2022Sep}
Let $\mathscr{M}$ be a semifinite von~Neumann algebra with a faithful semifinite normal trace $\tau$. If $x$ is a $\tau$-measurable self-adjoint operator on $L_2(\mathscr{M},\tau)$ and $a\in D(\pi(x)),\text{ } b\in D(\pi'(x)) $, then $xa,\text{ } bx\in L_{2}(\mathscr{M},\tau)$ and we obtain 
 \begin{align}\label{extended pi,pi'}
 \pi(x)(a)=xa,\quad \pi'(x)(b)=bx.
\end{align}
\end{proposition}
Note that the second equation in \eqref{extended pi,pi'} is not part of \cite[Proposition 4]{Luczak2022Sep}, but the authors mention that it is true and that its proof is similar to that of the first equation in ~\eqref{extended pi,pi'}.

Next, we present the formula of the relative modular operator with respect to two normal states on a semifinite von~Neumann algebra. 

\begin{theorem}\label{thm:Relative modular operator} \cite[Theorem $8'$]{Luczak2022Sep}
Let $\mathscr{M}$ be a semifinite von~Neumann algebra with a faithful semifinite normal trace $\tau$ and let $\phi,\text{ }\omega$ be normal states on $\mathscr{M}$. By Remark~\ref{existence of hphi,homega}, there exist vector representatives $\xi_\phi,\text{ } \xi_\omega \in L_{2}(\mathscr{M},\tau)_{+}$, corresponding to $\phi$, $\omega$ respectively, that satisfy Equation~\eqref{eq: equation for vector representative}. We define \begin{align}\label{xi omega tilde}
    \widetilde{\xi_{\omega}}:=w(\xi_{\omega}) \text{ where } w:[0,+\infty)\to [0,+\infty) \text{ is the funtion } w(\lambda)=\begin{cases} 
\frac{1}{\lambda} & \text{if }  \lambda>0\\
0 & \text{if }  \lambda=0
\end{cases}.
\end{align}
 Then, for the relative modular operator $\Delta_{\phi,\omega}$ we have
\begin{align}\label{formula of LPW for relative modular operator}
\Delta_{\phi,\omega}^{\frac{1}{2}}=\overline{ \pi(\xi_{\phi})\pi'(\widetilde{\xi_{\omega}})} \text{ .}\end{align}
\end{theorem}
The operator $w(\xi_\omega)$ in the previous theorem is defined via the Borel functional calculus (see Theorem~\ref{thm: Borel functional calculus}), i.e. if $\xi_\omega= \int_{[0,+\infty)} \lambda dE_{\omega} (\lambda) $ is the spectral decomposition of the positive self-adjoint operator $\xi_\omega$, then $w(\xi_\omega)=\int_{[0,+\infty)} w( \lambda) dE_{\omega} (\lambda)$.

The multiplicative form of the Spectral Theorem for strongly commuting normal operators is very important to our work. It is the key that allows us to create a bridge between the quantum and the classical states. Before we present it, we review the definition of  strongly commuting normal operators, using \cite[Proposition 5.27]{Schmudgen}.

\begin{definition}\label{defn: strongly commuting}
    Let $T_1,\text{ }T_2$ be normal operators on a Hilbert space $\mathscr{H}$ and let $E_{T_1}(\cdot),\text{ }E_{T_2}(\cdot) $ be their spectral measures. We say that $T_1,\text{ }T_2$ are {\bf{strongly commuting}} if their spectral measures commute, that is,
    \[ E_{T_1}(A)E_{T_2}(B)=E_{T_2}(B)E_{T_1}(A) \text{ for all Borel sets $A,\text{ }B \subseteq \mathbb{C}$}.\]
\end{definition}

\begin{theorem}\label{thm:Multiplication form of Spectral} \cite[Theorem 7.3]{BinGui} 
Let $T_1,\ldots,\text{ }T_{N}$ be strongly commuting normal operators (possibly unbounded) on a Hilbert space $\mathscr{H}$. Then, there exist a set $(\mu_{n})_{n\in I}$ of finite Borel measures on $\mathbb{C}^{N}$ and a unitary map $U:\mathscr{H}\to \oplus_{n\in I} L_2(\mathbb{C}^{N},\mu_n)$ such that
\[ UT_jU^{-1}=M_{z_j} \text{,   } \forall j=1,\ldots,\text{ }N  \]
where $z_j$ is the  map sending $(\zeta_1,\ldots,\text{ }\zeta_N)\in \mathbb{C}^{N}$ to $\zeta_j\in \mathbb{C}$ and $M_{z_j}$ denotes the multiplication operator by the function $z_j$, i.e. $M_{z_j} (y)=z_jy$ (the product of the two functions) for any $y\in \oplus_{n\in I} L_2(\mathbb{C}^{N},\mu_n)$. Recall that the domain of $M_{z_j}$ is the set of all $y\in \oplus_{n\in I} L_2(\mathbb{C}^{N},\mu_n)$ such that $z_jy \in \oplus_{n\in I} L_2(\mathbb{C}^{N},\mu_n)$. Moreover, if the operators $T_1,\ldots,$ $T_N$ are positive self-adjoint, then $\mathbb{C}^N$ can be replaced by $[0,+\infty)^N$.
\end{theorem}
In the setting of Theorem~\ref{thm:Relative modular operator}, we will apply Theorem~\ref{thm:Multiplication form of Spectral} for the operators $\pi(\xi_\phi)$, $\pi'(\xi_\omega)$. In the following lemma, we show that these operators are positive self-adjoint and strongly commuting. The proof of the lemma is rather obvious, but we present it for the sake of completeness. 

\begin{lemma}\label{remark: applying spectral theorem on pi of xi phi}
     Let $\mathscr{M}$ be a semifinite von~Neumann algebra with a faithful semifinite normal trace $\tau$ and let $\phi,\text{ }\omega$ be normal states on $\mathscr{M}$. Let  $\xi_\phi,\text{ }\xi_\omega \in L_{2}(\mathscr{M},\tau)_{+}$ be the  vector representatives corresponding to $\phi,\text{ }\omega$, respectively, satisfying Equation~\eqref{eq: equation for vector representative}.  The operators $T_1=\pi(\xi_{\phi}),\text{ } T_2=\pi'(\xi_{\omega})$  are  positive self-adjoint  strongly commuting operators. 
\end{lemma}

\begin{proof}
    Let \[\xi_\phi=\int_{[0,+\infty)} \lambda dE_{\phi}(\lambda),\quad \xi_\omega=\int_{[0,+\infty)} \lambda dE_{\omega}(\lambda)\] be the spectral decompositions of $\xi_\phi$, $\xi_\omega$  with spectral measures $E_\phi(\cdot),\text{ }E_\omega(\cdot)$ respectively. Since $\xi_\phi$, $\xi_\omega$ are self-adjoint operators affiliated with $\mathscr{M}$, it is true that $E_\phi(B)$, $E_\omega (B)\in \mathscr{M}$ for every Borel set $B\subseteq[0,+\infty)$. The measures  $\pi \circ E_\phi, \text{ }\pi' \circ E_\omega: Borel\left([0,+\infty)\right)\to \mathbb{B}\left(L_2(\mathscr{M},\tau)\right)$ are defined via Equations~\eqref{Eq: Left Multiplication} and \eqref{Eq: right multiplication}, and they  are the spectral measures of $\pi(\xi_\phi)$ and $\pi'(\xi_\omega)$, respectively. By Equation~\eqref{eq: pi,pi' of x}, we have 
\[\pi(\xi_\phi)=\int_{[0,+\infty)} \lambda \quad d(\pi\circ E_\phi)(\lambda), \quad \pi'(\xi_\omega)=\int_{[0,+\infty)} \lambda \quad d(\pi'\circ E_\omega)(\lambda). \]
Therefore, $\pi(\xi_\phi),\text{ }\pi'(\xi_\omega)$ are positive self-adjoint operators by the spectral theory. It is obvious that
\begin{align*} [(\pi\circ E_\phi)(A)(\pi'\circ E_\omega)(B) ](\xi)=E_\phi(A) \xi E_\omega(B)=
[(\pi'\circ E_\omega)(B)(\pi\circ E_\phi)(A) ](\xi),
\end{align*}
for all $A,\text{ }B$ Borel subsets of $[0,+\infty)$ and  $\xi \in L_2(\mathscr{M},\tau)$. Therefore, $\pi(\xi_\phi)\text{ and }\pi'(\xi_\omega)$ are strongly commuting by Definition~\ref{defn: strongly commuting}. 
\end{proof}

In the next theorem, we combine Theorem~\ref{thm:Relative modular operator} and Theorem~\ref{thm:Multiplication form of Spectral} to establish a formula for the relative modular operator in terms of multiplication operators. This is a key result in our work. Before we present the theorem, we recall that if $(X,\Sigma,\mu)$ is a measure space, then the commutative $L_2$-space  $L_2(X,\Sigma,\mu)$ (or simply $L_2(X,\mu)$) is defined as \[L_2(X,\mu):= \Bigl\{f:X\to \mathbb{C} : \int_{X}|f(x)|^2d\mu(x)<+\infty \Bigr\},  \]
where we identify functions that are equal $\mu$-almost everywhere. The space $L_2(X,\mu)$ is a Hilbert space with inner product $\braket{f}{g}_{L_2(X,\mu)}:=\int_{X} \overline{f(x)}g(x)d\mu(x)$ for all $f$, $g\in L_2(X,\mu)$. 

\begin{theorem}
    
\label{prop: application of the multiplication form of spectral theorem}
    Let $\mathscr{M}$ be a semifinite von~Neumann algebra with a faithful semifinite normal trace $\tau$ and  $\phi$, $\omega$ be normal states on $\mathscr{M}$. Let $\xi_\phi,\text{ }\xi_\omega \in L_{2}(\mathscr{M},\tau)_{+}$ be the vector representatives of $\phi,\text{ } \omega $ respectively, satisfying Equation~\eqref{eq: equation for vector representative}, and  $\Delta_{\phi,\omega}$ be the relative modular operator with respect to $\phi$ and $\omega$. Then, there exists a measure space $(X, \Sigma,\mu)$, a unitary map $U:L_2(\mathscr{M},\tau)\to L_2(X,\mu) $ and two Borel measurable functions  $g_\phi$, $g_\omega:X\to [0,+\infty)$ (where both functions as well as the unitary map $U$ depend on both states $\phi$ and $\omega$) such that
       \begin{align} \label{Eq: multiplication form we use}
U\pi(\xi_\phi)U^{-1}=M_{g_\phi}, \quad  U\pi' (\xi_\omega)U^{-1}=M_{g_\omega}  
\end{align}
\begin{align}\label{eq: square root of Delta mult oper}
\text{ and } \quad U\Delta_{\phi,\omega}^{\frac{1}{2}}U^{-1}=M_{g_{\phi}w(g_{\omega})},
\end{align}
where $w(\cdot)$ is the function defined in Equation~\eqref{xi omega tilde}.
\end{theorem}
Note that Equations~\eqref{Eq: multiplication form we use}, \eqref{eq: square root of Delta mult oper}  hold on the domains of the multiplication operators. Recall that if $M_g$ is the multiplication operator with respect to a Borel measurable function $g:X \to \mathbb{C}$, the domain of $M_g$ is the set \[\{h\in L_2(X,\mu): gh\in L_2(X,\mu)\}=\Bigl\{ h\in L_2(X,\mu): \int_{X} |g(x)|^2|h(x)|^2d\mu(x)<+\infty\Bigr\}.\]
\begin{proof}
      By Lemma~\ref{remark: applying spectral theorem on pi of xi phi}, the operators $T_1=\pi(\xi_\phi)  \text{ and } T_2=\pi'(\xi_\omega)$ are  positive self-adjoint   strongly commuting on the Hilbert space $L_2(\mathscr{M},\tau)$. By Theorem~\ref{thm:Multiplication form of Spectral},  there exist a set $(\mu_n)_{n\in I}$ of finite Borel measures on $[0,+\infty)^2$ and a unitary map  $U:L_2(\mathscr{M},\tau)\to \oplus_{n\in I} L_2([0,+\infty)^2,\mu_n)$ such that
      \begin{align} \label{multiplication form that we won't use}
U\pi(\xi_\phi)U^{-1}=M_{z_1} \text{  and  }  U\pi'(\xi_\omega)U^{-1}=M_{z_2}
\end{align}
where $z_j$ is the  map sending $(\zeta_1,\zeta_2)\in [0,+\infty)^2$ to $\zeta_j$ for $j=1$, $2$. We observe that  $\oplus_{n\in I} L_2([0,+\infty)^2,\mu_n)=L_2(X,\mu) \text{ where }$
\begin{itemize}
    \item $X := \sqcup_{ n \in I} [0,+\infty)^2=\cup_{n\in I}\left([0,+\infty)^2\times\{n\} \right)$,
    \item $ \mu:=\oplus_{n\in I}\mu_n \text{ is the measure defined as } 
 \mu(A):= \sum_{n\in I} \mu_{n}\big( A\cap ([0,+\infty)^2 \times \{ n \}) \big)$  for every  subset $A$ of $X$ in the smallest $\sigma$-algebra $\Sigma$ that contains the set \newline$ \cup_{n\in I} \left( Borel\left([0,+\infty)^2\right) \times \{n\} \right)$.
\end{itemize}

Next, we define the functions $g_\phi,\text{ }g_\omega: X\to[0,+\infty)$ as 
\begin{align} \label{E:the_formulas_of_g}
\begin{split}
    g_\phi( (\zeta_1,\zeta_2,n))&:=\zeta_1, \text{ for all } (\zeta_1,\zeta_2) \in [0,+\infty)^2  \text{ and } n\in I, \\
     g_\omega( (\zeta_1,\zeta_2,n))&:=\zeta_2, \text{ for all } (\zeta_1,\zeta_2) \in [0,+\infty)^2 \text{ and } n\in I.
     \end{split}
\end{align}
Thus, we can rewrite the unitary map $U$ as $U:L_2(\mathscr{M},\tau)\to L_2(X,\mu)$ and Equation~\eqref{multiplication form that we won't use} becomes $U\pi(\xi_\phi)U^{-1}=M_{g_\phi}$ and $ U\pi'(\xi_\omega)U^{-1}=M_{g_\omega}$.

Next, we focus on proving Equation~\eqref{eq: square root of Delta mult oper}.
We begin by applying the Borel measurable function $w$, which was introduced in Equation~\eqref{xi omega tilde}, on Equation~\eqref{Eq: multiplication form we use} to obtain
\begin{align}\label{eq: U of pi prime of xi tilde}
w(U\pi'(\xi_\omega)U^{-1})&=w(M_{g_\omega}) \notag\\
\Leftrightarrow                Uw(\pi'(\xi_{\omega}))U^{-1}&=M_{w(g_\omega)} \quad \text{ (by Lemma~\ref{lem: applying the Borel functional calculus} and \cite[Example 5.3]{Schmudgen})} \notag\\
\Leftrightarrow U\pi'\left(w(\xi_\omega)\right)U^{-1}&=M_{w(g_\omega)} \quad \text{ (by Equation~\eqref{Eq: pi of f of x}}) \notag\\
\Leftrightarrow U\pi'(\widetilde{\xi_\omega})U^{-1}&=M_{w(g_\omega)} \quad \text{(since } \widetilde{\xi_\omega}=w(\xi_\omega) \text{)}.\end{align}
Now that we have established Equation~\eqref{eq: U of pi prime of xi tilde}, we finish the proof of the theorem. 
\begin{align*}
   U\Delta_{\phi,\omega}^{\frac{1}{2}}U^{-1}&=U\overline{ \pi(\xi_{\phi})\pi'(\widetilde{\xi_{\omega}})}U^{-1} \text{ (by Theorem~\ref{thm:Relative modular operator})}\\ 
&=\overline{U  \pi(\xi_{\phi})\pi'(\widetilde{\xi_{\omega}})U^{-1}} \quad \text{(by \cite[Lemma 2]{Luczak2022Sep})}\\
&=\overline{U  \pi(\xi_{\phi})U^{-1} U\pi'(\widetilde{\xi_{\omega}})U^{-1}}=\overline{M_{g_\phi} M_{w(g_\omega)}} \quad (\text{by Equations}~\eqref{Eq: multiplication form we use},~\eqref{eq: U of pi prime of xi tilde} \text{) }\\
&=M_{g_\phi w(g_\omega)},
\end{align*} 
where the last equality is true by \cite[Theorem 4]{Nelson1974}. 
\end{proof}
We close this section by generalizing Theorem~\ref{prop: application of the multiplication form of spectral theorem} in the next corollary, which will be useful in Theorem~\ref{lem: first term} 
where we study the first term of the quantum $f$-divergence.

\begin{corollary}\label{thm:fDelta} 
Let $\mathscr{M}$ be a semifinite von~Neumann algebra,  $\phi$, $\omega$ be normal states on $\mathscr{M}$ and $\Delta_{\phi,\omega}$ be the relative modular operator with respect to $\phi$ and $\omega$. Let $f:[0,+\infty)\to\mathbb{R}$ be a Borel measurable function. Let $U$ be the unitary map  and $g_\phi,\text{ }g_\omega: X\to [0,+\infty)$ be the functions that appear in  Theorem~\ref{prop: application of the multiplication form of spectral theorem}. If $f_\phi:=g_\phi^2$ and $f_\omega:=g_\omega^2$, then we obtain
\begin{align}\label{eq1:f(Delta phi,omega)}
Uf(\Delta_{\phi,\omega})U^{-1}=M_{f(f_{\phi}w(f_{\omega}))} 
\end{align}
where the equality holds on the domain of the multiplication operator, i.e.  $\{g\in L_2(X,\mu): f\left(f_{\phi}w(f_{\omega})\right)(\cdot)g(\cdot) \in L_2(X,\mu) \}. $
\end{corollary}
\begin{proof}
First, we define the  function $K:\mathbb{R}\to \mathbb{R}$ by the formula $K(t):=f(t^2)$. Then,
\begin{align*}
U\Delta_{\phi,\omega}^{\frac{1}{2}}U^{-1}&=M_{g_\phi w(g_\omega)} \quad \text{ (by Theorem~\ref{prop: application of the multiplication form of spectral theorem}) }\\
  \Rightarrow  K(U\Delta_{\phi,\omega}^{\frac{1}{2}}U^{-1})&=K(M_{g_\phi w(g_\omega)}) \quad \ \text{ (by applying the function } K \text{)}\\
    \Leftrightarrow U K(\Delta_{\phi,\omega}^{\frac{1}{2}}) U^{-1}&=M_{K(g_\phi w(g_\omega))} \quad\text{ (by Lemma~\ref{lem: applying the Borel functional calculus} and \cite[Example 5.3]{Schmudgen})}  \\
    \Leftrightarrow Uf(\Delta_{\phi,\omega})U^{-1}&=M_{f(g_\phi^2 w^2(g_\omega))} \quad \text{ (since } K(t)=f(t^2) \text{)} \\
    &=M_{f(f_\phi w(g_\omega^2))} \quad \text{ (since } f_\phi:=g_\phi^2 \text{ and the functions }  w,t^2 \text{ commute)}\\
    &=M_{f(f_\phi w(f_\omega))} \quad \text{ (since } f_\omega:=g_\omega^2 \text{)}.
\end{align*}
\end{proof}

\section{The Relation between Quantum and Classical  \texorpdfstring{\ensuremath{f}}{}-divergences via Nussbaum-Szko{\l}a Distributions}\label{sec:f-divergence} 
Our goal in this section is to prove our main result, namely Theorem~\ref{maintheorem}. The proof is provided towards the end of this section. We begin  with some standard notation and conventions related to a convex  function $f:(0,+\infty)\to \mathbb{R}$. We fix the notation
 \begin{align}\label{eq:convex-f-notations-1}
    f(0^{+}):=\lim_{t\to 0^{+}} f(t) \text{ and } f^{'}(+\infty):=\lim_{t\to +\infty} \frac{f(t)}{t} \text{ .}
\end{align}
Both $f(0^{+})$ and $f^{'}(+\infty)$ exist in $(-\infty,+\infty]$ due to the fact that $f$ is convex. We also fix the common conventions
\begin{align}\label{eq:convex-f-notations-2}
   \begin{split}
      & 0f\left(\frac{0}{0}\right) := 0,\text{ } 0\cdot (\pm\infty):= 0, \text{ } a \cdot (\pm\infty) := \pm\infty \text{ for  } a >0,\\
    & 0f\left(\frac{a}{0}\right) :=  \lim\limits_{t\rightarrow 0^{+}}tf\left(\frac{a}{t}\right) = a\lim\limits_{s\rightarrow +\infty}\frac{f(s)}{s}=af'(+\infty)  \text{ for  } a>0.\\
   \end{split}
 \end{align} 
Next, we define the classical $f$-divergence following \cite{Liese-Vajda-2006}. The notion was first introduced by Csisz{\ifmmode\acute{a}\else\'{a}\fi}r in \cite{Csiszar1963}.
 \begin{definition}
 \label{defn:classical f divergence} Let $P$,  $ Q$ be probability measures on a measurable space $(X,\Sigma)$. Let $\nu$ be any $\sigma$-finite measure such that $P\ll\nu$ and $Q\ll\nu$ (i.e. if $\nu(A)=0$ for some $A\in \Sigma$, then $P(A)=Q(A)=0$). Let $p:=\frac{dP}{d\nu}$ and $q:=\frac{dQ}{d\nu}$  denote the Radon-Nikodym derivatives of $P$ and $Q$ respectively, with respect to $\nu$,  and let $f:(0,+\infty)\to \mathbb{R}$ be a convex function. The {\bf{ classical \boldmath{\texorpdfstring{\ensuremath{f}}{}}-divergence}} $D_f(P\|Q)$ between the classical states $P$ and $Q$  is defined in \cite[Equation (24)]{Liese-Vajda-2006} as
  \[D_f(P\|Q) := \int\limits_{X}qf\left(\frac{p}{q}\right)d\nu\]
 using the  notation in~\eqref{eq:convex-f-notations-1} and the conventions in~\eqref{eq:convex-f-notations-2}. If we expand this formula, we obtain   
 \begin{align}\label{eq: formula of the classical divergence}
 D_f(P\|Q)= \int\limits_{\{p,q>0\}}qf\left(\frac{p}{q}\right)d\nu+ f(0^{+})\int\limits_{\{p=0,q>0\}}qd\nu+f^{'}(+\infty)\int\limits_{\{p>0,q=0\}}pd\nu 
 \end{align}
where we have adopted a common notation for sets in measure theory, namely for example $\{p,q>0\}$ denotes the set $\{x\in X:p(x)>0 \text{ and } q(x)>0\}$. 
\end{definition}
Notice that the set $\{p=0,q=0\}$ does not contribute to the \texorpdfstring{\ensuremath{f}}{}-divergence and this is where the convention $0f\left(\frac{0}{0}\right) := 0$ is used. 

We continue with the definition of the quantum \texorpdfstring{\ensuremath{f}}{}-divergence between two normal positive functionals on a general von~Neumann algebra, following \cite{Hiai-2018}.

\begin{definition}\cite[Definition 2.1]{Hiai-2018}\label{defn:f-divergence}
 Let $(\mathscr{M},\mathscr{H},J,P)$ be a standard form of a von~Neumann algebra. Let $\phi$, $\omega \in \mathscr{M}_{*}^{+}$ and  $f:(0,+\infty)\to \mathbb{R}$ be a convex  function. Recall from Definition~\ref{relative modular operator} the notion of  the relative modular operator $\Delta_{\phi,\omega}$ and recall that its support projection is $s_\mathscr{M}(\phi)s_{\mathscr{M}'}(\omega)$. We write the spectral decomposition of $\Delta_{\phi,\omega}$ as
$ \Delta_{\phi,\omega}=\int_{[0,+\infty)} t dE_{\phi,\omega}(t)$ 
 and  define the operator $f( \Delta_{\phi,\omega})$ on $s_\mathscr{M}(\phi)s_{\mathscr{M}'} (\omega)\mathscr{H}$  by the formula 
 \[\label{f of Delta}
     f(\Delta_{\phi,\omega}):= \int_{(0,+\infty)} f(t) dE_{\phi,\omega}(t).\]
 We define the {\bf{ quantum \boldmath{\texorpdfstring{\ensuremath{f}}{}}-divergence} } $S_f(\phi\|\omega)$ between $\phi$ and $\omega$ as
\begin{align}\label{eq: formula of the quantum f divergence}
    S_f(\phi\|\omega):=\langle\langle\xi_\omega|f(\Delta_{\phi,\omega})(\xi_\omega) \rangle\rangle+f(0^{+})\omega(1-s_{\mathscr{M}}(\phi))+ f^{'}(+\infty)\phi(1-s_{\mathscr{M}}(\omega)) 
\end{align}
 using the notation in~\eqref{eq:convex-f-notations-1} where the term $ \langle\langle\xi_\omega|f(\Delta_{\phi,\omega})(\xi_\omega)\rangle\rangle$ is defined as
 \begin{align}\label{eq: first term of divergence}
 \langle\langle\xi_\omega|f(\Delta_{\phi,\omega})(\xi_\omega) \rangle\rangle:=\int_{(0,+\infty)}f(t)d\|E_{\phi,\omega}(t)\xi_\omega\|_{\mathscr{H}}^2.
\end{align}
 \end{definition}
Notice that the term $\langle\langle\xi_\omega|f(\Delta_{\phi,\omega})(\xi_\omega) \rangle\rangle$ is not necessarily an inner product because $\xi_\omega$ may not belong to the domain of $f(\Delta_{\phi,\omega})$. It should be understood as the integral on the right-hand side of Equation~\eqref{eq: first term of divergence}. We use the notation $\langle\langle \cdot|\cdot \rangle\rangle$ to avoid any confusion with the inner products that will appear  in our proofs. 

By \cite[Lemma 2.2]{Hiai-2018}, $S_f(\phi\|\omega)$ is well defined with values in $(-\infty,+\infty]$, and by \cite[Proposition 2.3 (1)]{Hiai-2018} $S_f(\phi\|\omega)$ is independent of the standard form of the von~Neumann algebra we choose to use. 

In the next theorem, we discover a useful formula for the first term of the quantum $f$-divergence.

\begin{theorem}\label{lem: first term}
Let $\mathscr{M}$ be a semifinite von~Neumann algebra with a faithful semifinite normal trace $\tau$. Let $\phi$, $\omega$ be normal states on $\mathscr{M}$ and  $\xi_\omega \in L_2(\mathscr{M}, \tau)_+$ 
be the vector representative of $\omega$. Let $\Delta_{\phi,\omega}$ be the relative modular operator with respect to $\phi$ and $\omega$ and  $E_{\phi,\omega}(\cdot)$ be its spectral measure. Let $\mu$ be the measure and $U$ be the unitary map which appear in Theorem~\ref{prop: application of the multiplication form of spectral theorem}. Let $f_\phi$, $f_\omega$ be the two functions that were  obtained in Corollary~\ref{thm:fDelta} and  $f:(0,+\infty)\to \mathbb{R}$ be a convex function. Then,
    \begin{align}\label{eq: first term analyzed} \langle\langle\xi_\omega|f(\Delta_{\phi,\omega})(\xi_\omega) \rangle\rangle =\int_{\{f_\phi,f_\omega>0\}} f\left(\frac{f_\phi}{f_\omega}\right) |U(\xi_\omega)|^2d\mu.
    \end{align}
\end{theorem}

\begin{proof}
    Let $f^{+}(t):= \max\{f(t),0\}\text{ and }f^{-}(t):= \max\{-f(t),0\}$ for all $t>0$ be the positive and the negative part of the function $f$, respectively. We recall that $f^{+}(t),\text{ }f^{-}(t)\geq 0$ and $f(t)=f^{+}(t)-f^{-}(t)$ for all $t>0$.  We will prove Equation~\eqref{eq: first term analyzed} starting from the right-hand side.

\begin{align}\label{eq: 4.3 first equation of the proof}
 \int_{\{f_\phi, f_\omega>0\}}& f \left(\frac{f_\phi}{f_\omega}\right) |U(\xi_\omega)|^2d\mu= \int_{\{f_\phi, f_\omega>0\}} \left(f^{+}-f^{-}\right) \left(\frac{f_\phi}{f_\omega}\right) |U(\xi_\omega)|^2d\mu \notag \\
      &=\int_{\{f_\phi, f_\omega>0\}} f^{+} \left(\frac{f_\phi}{f_\omega}\right) |U(\xi_\omega)|^2d\mu-\int_{\{f_\phi, f_\omega>0\}} f^{-} \left(\frac{f_\phi}{f_\omega}\right) |U(\xi_\omega)|^2d\mu.
\end{align}
We would like to explain why the last expression makes sense by proving that \begin{align}\label{eq: 4.3 second equation of the proof}
    \int_{\{f_\phi, f_\omega>0\}} f^{-} \left(\frac{f_\phi}{f_\omega}\right) |U(\xi_\omega)|^2d\mu <+\infty.
\end{align}
If that is true, the indeterminate form $\left( +\infty\right)-\left( +\infty\right)$ is avoided and the right-hand side of Equation~\eqref{eq: 4.3 first equation of the proof} is well defined in $(-\infty,+\infty]$. Due to the convexity of the function $f$,  there exist $\alpha$, $\beta \in \mathbb{R}$ such that $f(t)\geq \alpha+\beta t$ or equivalently $-f(t)\leq -\alpha-\beta t$  for all $t>0$.    Therefore, 
\[f^{-}(t)\leq |-\alpha-\beta t|\leq |\alpha|+|\beta|t \quad\text{ for all } t>0. \]
We apply this inequality on the integral $\int_{\{f_\phi, f_\omega>0\}} f^{-} \left(\frac{f_\phi}{f_\omega}\right) |U(\xi_\omega)|^2d\mu$. 
\begin{align}\label{eq: lem4.8 equation}
    \int_{\{f_\phi, f_\omega>0\}} &f^{-} \left(\frac{f_\phi}{f_\omega}\right) |U(\xi_\omega)|^2d\mu\leq \int_{ {\{f_\phi, f_\omega>0\}} } \left(|\alpha|+|\beta|\frac{f_\phi}{f_\omega}\right) |U(\xi_\omega)|^2d\mu \notag\\
    &=|\alpha|\int_{\{f_\phi,f_\omega >0\}} |U(\xi_\omega)|^2d\mu +|\beta|\int_{\{f_\phi,f_\omega >0\}} \frac{f_\phi}{f_\omega}|U(\xi_\omega)|^2d\mu \notag\\
    &\leq |\alpha|\int_{X} |U(\xi_\omega)|^2d\mu + |\beta|\int_{\{f_\phi,f_\omega >0\}}  \frac{f_\phi}{f_\omega}|U(\xi_\omega)|^2d\mu.
\end{align}
To verify Inequality~\eqref{eq: 4.3 second equation of the proof}, it suffices to show that the two integrals on the right-hand side of  Inequality~\eqref{eq: lem4.8 equation}  are finite. For the first integral, we use the fact that $U$ is unitary and Remark~\ref{existence of hphi,homega} to obtain
\[\int_{X} |U(\xi_\omega)|^2d\mu=\braket{U(\xi_\omega)}{U(\xi_\omega)}_{L_2(X,\mu)}=\braket{\xi_\omega}{\xi_\omega}_{L_2(\mathscr{M},\tau)}=\omega(1)=1<+\infty.\]
For the second integral, we begin by observing that $\xi_\omega $ belongs to the domain of the operator $S_{\phi,\omega}$, which was introduced in Equation~\eqref{eq: S of phi,omega}. Hence, $\xi_\omega$ also belongs to the domain of the closure $\overline{S_{\phi,\omega}}$. Since $\overline{S_{\phi,\omega}}=J\Delta_{\phi,\omega}^{\frac{1}{2}}$, by Equation~\eqref{eq: polar decomp of S}, we deduce that $\xi_\omega $ belongs to the domain of $\Delta_{\phi,\omega}^{\frac{1}{2}}$. Therefore, $U(\xi_\omega)$ belongs to the domain of the operator $U\Delta_{\phi,\omega}^{\frac{1}{2}}U^{-1}$ and by Corollary~\ref{thm:fDelta}, applied for the square root function, the operator $U\Delta_{\phi,\omega}^{\frac{1}{2}}U^{-1}$ is equal to the multiplication operator $M_{\sqrt{f_\phi w(f_\omega)}}$. Thus, $U(\xi_\omega)$ belongs to the domain of the operator $M_{\sqrt{f_\phi w(f_\omega)}}$. Equivalently, the product $\sqrt{f_\phi w(f_\omega)}(\cdot) U(\xi_\omega)(\cdot)$ belongs to $L_2(X,\mu)$. Hence, 
\begin{align*}
&\int_{X} \left| \sqrt{f_\phi w(f_\omega)} U(\xi_\omega)\right|^2d\mu <+\infty \\
\Leftrightarrow & \int_{X} f_\phi w(f_\omega) |U(\xi_\omega)|^2d\mu=\int_{\{f_\phi, f_\omega>0\}} f_\phi w(f_\omega) |U(\xi_\omega)|^2d\mu<+\infty\\
\Leftrightarrow &\int_{\{f_\phi, f_\omega>0\}} \frac{f_\phi}{f_\omega}|U(\xi_\omega)|^2d\mu<+\infty \quad \text{(by Equation~\eqref{xi omega tilde})}.
\end{align*}
Therefore, both integrals of Inequality~\eqref{eq: lem4.8 equation} are finite and this verifies Inequality~\eqref{eq: 4.3 second equation of the proof}, which in turn   verifies that Equation~\eqref{eq: 4.3 first equation of the proof} makes sense as a number in $(-\infty,+\infty]$.

Next, we define  the functions $f_n,\text{ }g_n:[0,+\infty) \to \mathbb{R}$ by
\[f_n(t):=\begin{cases} 
f^{+}(t)\mathscr{X}_{ \{f^{+}\leq n\} }(t) & \text{if }  t>0\\
\quad 0 & \text{if }  t=0
\end{cases} \quad \text{and}  \quad g_n(t):=\begin{cases} 
f^{-}(t)\mathscr{X}_{ \{f^{-}\leq n\} }(t) & \text{if }  t>0\\
\quad 0 & \text{if }  t=0
\end{cases} \]
for all $n\in \mathbb{N},$  $\text{ } t>0$ where by $\mathscr{X}_A(\cdot)$ we denote the characteristic function of any set $A$. We observe that  $0\leq f_1(t)\leq f_2(t)\leq \ldots$ and $\lim_{n\to +\infty} f_{n}(t)=f^{+}(t)$ for all $t>0$. Similarly, $0\leq g_1(t)\leq g_2(t)\leq \ldots$ and $\lim_{n\to +\infty} g_{n}(t)=f^{-}(t)$ for all $t>0$. 
We return to Equation~\eqref{eq: 4.3 first equation of the proof} and apply the Monotone Convergence Theorem to obtain
\begin{align}\label{eq: 4.3 third equation of the proof}
 &\int_{\{f_\phi, f_\omega>0\}} f \left(\frac{f_\phi}{f_\omega}\right) |U(\xi_\omega)|^2d\mu \notag \\ 
      &\quad\quad =\lim_{n\to +\infty} \left( \int_{\{f_\phi, f_\omega>0\}} f_n \left(\frac{f_\phi}{f_\omega}\right) |U(\xi_\omega)|^2d\mu - \int_{\{f_\phi, f_\omega>0\}} g_n \left(\frac{f_\phi}{f_\omega}\right) |U(\xi_\omega)|^2d\mu \right) .
\end{align}
 Next, we fix some $n\in \mathbb{N}$ and focus on the integral $\int_{\{f_\phi, f_\omega>0\}} f_n \left(\frac{f_\phi}{f_\omega}\right) |U(\xi_\omega)|^2d\mu$. 
 
 \begin{align}\label{eq: 4.3 fourth equation of the proof}
  \int_{\{f_\phi, f_\omega>0\}} f_n \left(\frac{f_\phi}{f_\omega}\right) |U(\xi_\omega)|^2d\mu&=\int_{X} \overline{U(\xi_\omega)} f_n \left(f_\phi w(f_\omega)\right) U(\xi_\omega)d\mu \notag\\ &\text{(by Equation~\eqref{xi omega tilde} and the fact that  } f_n(0)=0 \text{)} \notag \\
  &=\braket{U(\xi_\omega)}{M_{f_n(f_\phi w(f_\omega))}U(\xi_\omega)}_{L_2(X,\mu)} \notag \\
  &=\braket{\xi_\omega}{U^{-1}M_{f_n(f_\phi w(f_\omega))}U(\xi_\omega)}_{L_2(\mathscr{M},\tau)}.
  \end{align}
Since $f_n$ is a bounded function, the operator $f_n(\Delta_{\phi,\omega})$ is bounded (by the bounded functional calculus). Thus, $\xi_\omega$ belongs to the domain of $f_n(\Delta_{\phi,\omega})$. Equivalently, $U(\xi_\omega)$ belongs to the domain of $Uf_{n}(\Delta_{\phi,\omega})U^{-1}$. Thus, we can apply Corollary~\ref{thm:fDelta} to obtain
   \begin{align}\label{eq: 4.3 fifth equation of the proof}
       Uf_{n}(\Delta_{\phi,\omega})U^{-1} U(\xi_\omega)&=M_{f_n(f_\phi w(f_\omega))}U(\xi_\omega) \notag\\
       \Leftrightarrow f_{n}(\Delta_{\phi,\omega})(\xi_\omega)&=U^{-1}M_{f_n(f_\phi w(f_\omega))}U(\xi_\omega).
   \end{align}
   We combine Equations~\eqref{eq: 4.3 fourth equation of the proof}, ~\eqref{eq: 4.3 fifth equation of the proof} and infer
   \begin{align}\label{eq: 4.3 sixth equation of the proof}
        \int_{\{f_\phi, f_\omega>0\}} f_n \left(\frac{f_\phi}{f_\omega}\right) |U(\xi_\omega)|^2d\mu&=\braket{\xi_\omega}{f_{n}(\Delta_{\phi,\omega})(\xi_\omega)}_{L_2(\mathscr{M},\tau)}    \notag \\   &=\int_{[0,+\infty)} f_n(t) d\braket{\xi_\omega}{E_{\phi,\omega}(t)\xi_\omega}_{L_2(\mathscr{M},\tau)}   \notag\\
       &= \int_{[0,+\infty)} f_n(t) d\|{E_{\phi,\omega}(t)\xi_{\omega}}\|_2^2 \notag\\
       &=\int_{(0,+\infty)} f_n(t) d\|{E_{\phi,\omega}(t)\xi_{\omega}}\|_2^2 
   \end{align}
   where in the second equality we used a property of the functional calculus (see  \cite[Theorem 5.9 (1)]{Schmudgen}), in the third equality we used the fact that $E_{\phi,\omega} (\cdot)$ is a projection-valued measure  and in the last equality we used the fact that $f_n(0)=0$.

In exactly the same way, one can prove that  \begin{align}\label{eq: 4.3 seventh equation of the proof}
   \int_{\{f_\phi, f_\omega>0\}} g_n \left(\frac{f_\phi}{f_\omega}\right) |U(\xi_\omega)|^2d\mu= \int_{(0,+\infty)} g_n(t) d\|E_{\phi,\omega}(t)\xi_{\omega}\|_2^2.
\end{align}

 Now, we return to Equation~\eqref{eq: 4.3 third equation of the proof} and apply Equations~\eqref{eq: 4.3 sixth equation of the proof}, ~\eqref{eq: 4.3 seventh equation of the proof} to obtain
 \begin{align}
 \int_{\{f_\phi, f_\omega>0\}}&f \left(\frac{f_\phi}{f_\omega}\right) |U(\xi_\omega)|^2d\mu \nonumber \\
 &=\lim_{n\to +\infty} \left(  \int_{(0,+\infty)} f_n(t) d\|E_{\phi,\omega}(t)\xi_{\omega}\|_2^2 -  \int_{(0,+\infty)} g_n(t) d\|E_{\phi,\omega}(t)\xi_{\omega}\|_2^2 \right) \nonumber \\
 &=\int_{(0,+\infty)}f^{+}(t)d\|E_{\phi,\omega}(t)\xi_{\omega}\|_2^2-\int_{(0,+\infty)}f^{-}(t)d\|E_{\phi,\omega}(t)\xi_{\omega}\|_2^2,
 \label{E:4145}
 \end{align}
by the Monotone Convergence Theorem. Note that the
expression in Equation~\eqref{E:4145} makes sense since by 
Equation~\eqref{eq: 4.3 seventh equation of the proof}, the Monotone Convergence Theorem and Equation~\eqref{eq: 4.3 second equation of the proof} the indeterminate form $\infty - \infty$ cannot occur.
 Then, Equation~\eqref{E:4145} becomes
 \begin{align*}
 \int_{\{f_\phi, f_\omega>0\}}&f \left(\frac{f_\phi}{f_\omega}\right) |U(\xi_\omega)|^2d\mu \nonumber \\
 &=\int_{(0,+\infty)} \left(f^{+}-f^{-}\right)(t)d\|E_{\phi,\omega}(t)\xi_{\omega}\|_2^2=\int_{(0,+\infty)}f(t)d\|E_{\phi,\omega}(t)\xi_{\omega}\|_2^2\\
 &= \langle\langle\xi_\omega|f(\Delta_{\phi,\omega})(\xi_\omega) \rangle\rangle \quad \text{(by Equation~\eqref{eq: first term of divergence})},
\end{align*}
which finishes the proof of the theorem.
\end{proof}

\begin{remark}\label{defining the measure nu}
     Let $\mathscr{M}$ be a semifinite von~Neumann algebra with a faithful semifinite normal trace $\tau$. Let   $\phi$, $\omega$ be normal states on $\mathscr{M}$, 
     $\xi_\phi,\ \xi_\omega \in L_2(\mathscr{M}, \tau)_+$ be their respective vector representatives  and  $f_\phi$, $f_\omega$ be the  two functions that were obtained in  Corollary~\ref{thm:fDelta}. Let $P_\phi$, $P_\omega$ be the measures that have densities $f_\phi$ and $f_\omega$ (i.e $f_\phi=\frac{dP_\phi}{d\nu}$ and $f_\omega=\frac{dP_\omega}{d\nu}$), respectively, with respect to the measure $\nu$ which is defined as
\begin{align}\label{eq:measure d nu}
  d\nu:= \begin{cases} 
\frac{|U(\xi_{\omega})|^2}{f_\omega}d\mu & \text{on the set }  \{f_{\omega}> 0\}\\
\frac{|U(\xi_{\phi})|^2}{f_\phi}d\mu & \text{on the set }  \{f_{\phi}> 0,f_{\omega}=0\}\\
 \quad 0 & \text{on the set }  \{f_{\phi}=f_{\omega}=0\}
\end{cases}
\end{align}
where $U$ is the unitary map and $\mu $ is the measure obtained in Theorem~\ref{prop: application of the multiplication form of spectral theorem}.  In our main result, Theorem~\ref{maintheorem}, we will show that for every convex function $f$ the quantum \texorpdfstring{\ensuremath{f}}{}-divergence between $\phi$ and $\omega$ is equal to the classical \texorpdfstring{\ensuremath{f}}{}-divergence between $P_\phi$ and $P_\omega$.
\end{remark}

Our next goal is to prove that $P_\phi$ and $P_\omega$ are classical states. Part of this proof requires the following lemma in which we analyze the  projection $s_\mathscr{M}(\omega)$, which was introduced in Definition~\ref{relative modular operator}, in the case of a semifinite von~Neumann algebra $\mathscr{M}$ where $\omega$ is a normal state on $\mathscr{M}$. This projection is also part of the formula of the quantum $f$-divergence (see  Equation~\eqref{eq: formula of the quantum f divergence}).

\begin{lemma}\label{lem: s M of omega}
  \sloppy  Let $\mathscr{M}$ be a semifinite von~Neumann algebra with a faithful  semifinite normal trace $\tau$,  $\omega$ be a normal state on $\mathscr{M}$ and $\xi_\omega \in L_{2}(\mathscr{M},\tau)_{+}$ be the vector representative of $\omega$. Let $s_\mathscr{M}(\omega)$ be the support of the state $\omega$ on $\mathscr{M}$, which was introduced in Definition~\ref{relative modular operator}. Let $\xi_\omega= \int_{[0,+\infty)} \lambda dE_{\omega} (\lambda)$ be the spectral decomposition of $\xi_\omega$, where $E_{\omega}(\cdot)$ is its spectral measure.  Then, $s_\mathscr{M}(\omega)$ is equal to $E_{\omega} \left( (0,+\infty)\right)$. 
\end{lemma}

\begin{proof}
 Since $E_{\omega}(\cdot)$ is a projection-valued measure, $E_{\omega} \left( (0,+\infty)\right)$ is a projection. Moreover, $E_{\omega} \left( (0,+\infty)\right) \in \mathscr{M}$ because $\xi_\omega$ is affiliated with $\mathscr{M}$. Hence, $1-E_{\omega}\left( (0,+\infty) \right)$ is a projection  that belongs to $\mathscr{M}$ and 
\begin{align*}
    \omega(1-E_{\omega}\left( (0,+\infty) \right))&=\braket{\xi_\omega}{\xi_\omega}-\braket{\xi_\omega}{E_{\omega}\left( (0,+\infty) \right) \xi_\omega} \quad \text{(by Remark \ref{existence of hphi,homega})}\\
    &=\braket{\xi_\omega}{\xi_\omega}-\braket{\xi_\omega}{\xi_\omega}=0 \quad \text{(by Lemma~\ref{lem:E zero infty})}.    
\end{align*}
Therefore, 
\begin{align}\label{eq: second equation of lemma 4.5}
    &1-E_{\omega}\left( (0,+\infty) \right)\leq \sup\{e\in \mathscr{M}_{proj}:\omega(e)=0\} \notag\\
    \Leftrightarrow\quad  &1-\sup\{e\in \mathscr{M}_{proj}:\omega(e)=0\}\leq E_{\omega}\left( (0,+\infty) \right)\notag\\
     \Leftrightarrow\quad  &s_{\mathscr{M}}(\omega)\leq E_{\omega}\left( (0,+\infty) \right) \quad \text{(by Definition~\ref{relative modular operator})}.
\end{align}
Next, we focus on proving the reverse inequality. We begin by proving that $s_{\mathscr{M}}(\omega)\xi_\omega=\xi_\omega$. By \cite[Section 5.15]{StratilaZsido2019}, we recall that $\omega\left(1-s_{\mathscr{M}}(\omega)\right)=0$. Hence,
\begin{align}\label{eq: lem 4.5 equation 3}
    \|\xi_\omega\|_{2}^2&=\braket{\xi_\omega}{\xi_\omega}=\omega(1)=\omega\left( s_{\mathscr{M}} (\omega)\right)\notag \\
    &=\braket{\xi_\omega}{s_{\mathscr{M}} (\omega)\xi_\omega}=\braket{\xi_\omega}{\pi \left(s_{\mathscr{M}} (\omega)\right)(\xi_\omega)}   \quad \text{(by Remark~\ref{existence of hphi,homega} and Equation~\eqref{Eq: Left Multiplication})}\notag\\
    &=\| \pi \left(s_{\mathscr{M}} (\omega)\right)(\xi_\omega) \|_{2}^2 \quad \text{(since $ \pi \left(s_{\mathscr{M}} (\omega)\right)$ is a projection on $L_2(\mathscr{M},\tau)$)}.
\end{align}
\sloppy For the complementary projection $1-\pi\left( s_{\mathscr{M}} (\omega)\right)$ of $\pi \left(s_{\mathscr{M}} (\omega)\right)$ we observe that  
$1-\pi\left( s_{\mathscr{M}} (\omega)\right)= \pi\left( 1-s_{\mathscr{M}}(\omega)\right)$ and obtain
$$
\|\xi_\omega\|_{2}^2= \|\pi\left( s_{\mathscr{M}}(\omega)\right) (\xi_\omega) \|_{2}^2+\|\pi \left(1-s_{\mathscr{M}}(\omega)\right) (\xi_\omega) \|_{2}^2,
$$ 
which combined with Equation~\eqref{eq: lem 4.5 equation 3} implies     that  $\|\pi\left(1-s_{\mathscr{M}}(\omega)\right)(\xi_\omega)\|_{2}^2=0$. Hence, $\pi\left(1-s_{\mathscr{M}}(\omega)\right)(\xi_\omega)=\left(1-s_{\mathscr{M}}(\omega)\right)\xi_\omega=0$. 
Thus, $s_{\mathscr{M}}(\omega)\xi_\omega=\xi_\omega$. By Lemma~\ref{lem:E zero infty}, $E_{\omega} \left( (0,+\infty)\right)$ is the smallest projection that satisfies $E_{\omega} \left( (0,+\infty)\right)\xi_\omega=\xi_\omega$. Therefore, 
\begin{align*}
  &s_{\mathscr{M}}(\omega)\geq  E_{\omega} \left( (0,+\infty)\right)\\
  \Rightarrow \quad &s_{\mathscr{M}}(\omega)=  E_{\omega} \left( (0,+\infty)\right) \quad \text{(by Inequality~\eqref{eq: second equation of lemma 4.5})}.
\end{align*}

\end{proof}

  \begin{proposition}\label{prop: f phi,omega are distributions} 
Let $\mathscr{M}$ be a semifinite von~Neumann algebra with a faithful semifinite normal trace $\tau$,   $\phi$, $\omega$ be normal states on $\mathscr{M}$ and  $\xi_\phi$, $\xi_\omega \in L_{2}(\mathscr{M},\tau)_{+}$ be the vector representatives of $\phi$, $\omega$ respectively. Let $U$ be the unitary map and $(X, \Sigma,\mu)$ be the measure space we obtained in Theorem~\ref{prop: application of the multiplication form of spectral theorem}. Let $f_\phi$, $f_\omega$ be the two functions we obtained in Corollary~\ref{thm:fDelta} and  $\nu$ be the measure we  defined in Equation~\eqref{eq:measure d nu}. Then, $f_\phi (x)$, $f_\omega(x) \geq 0$ for all $x\in X$ and $\int_{X}f_\phi d\nu=\int_X f_\omega d\nu=1$. In other words, the measures $P_\phi$ and $P_\omega$ with densities $f_\phi$ and $f_\omega$ (respectively) with respect to the measure $\nu$ are classical states on ${X}$.
\end{proposition}
\begin{proof} 
Recall by  Corollary~\ref{thm:fDelta} that $f_\phi=g_\phi^2 \text{ and }f_\omega=g_\omega^2$ where $g_\phi,\text{ }g_\omega$ are real-valued functions (in particular nonnegative), by Theorem~\ref{prop: application of the multiplication form of spectral theorem}. Therefore, $f_\phi (x)$, $f_\omega(x) \geq 0$ for all $x\in X$.

Next, we focus on proving that $\int_{X}f_\omega d\nu=1$. We begin by proving the following claim. \newline
{\bf{\underline{Claim 1:}}}
\begin{align}\label{U of x omega zero at f omega zero}
    \int_{\{f_\omega=0\}} |U(\xi_\omega)|^2d\mu=0
    \text{ and } \int_{\{f_\phi=0\}} |U(\xi_\phi)|^2d\mu=0.
\end{align}
\sloppy
 Let $\mathscr{X}_{\{f_\omega=0\}}$ be the characteristic function of the set $\{f_\omega=0\}$. We observe that $U(\xi_\omega)\mathscr{X}_{\{f_\omega=0\}} \in L_2(X,\mu)$. Indeed,
 \[\int_X  |U(\xi_\omega)\mathscr{X}_{\{f_\omega=0\}}|^2 d\mu=\int_{\{f_\omega=0\}}  |U(\xi_\omega)|^2d\mu\leq \int_{X}  |U(\xi_\omega)|^2d\mu <\infty \]
where the last inequality is true since $\xi_\omega \in L_2(\mathscr{M},\tau)$ and $U:L_2(\mathscr{M},\tau) \to L_2(X,\mu)$. In addition, we observe that the product $g_\omega U(\xi_\omega)\mathscr{X}_{\{f_\omega=0\}} $ is also in $ L_2(X,\mu)$ (since it is equal to zero). Therefore, $U(\xi_\omega)\mathscr{X}_{\{f_\omega=0\}}$ belongs to the domain of the multiplication operator $M_{g_\omega} $. Hence, we can apply Equation~\eqref{Eq: multiplication form we use} to obtain
\begin{align*}
U\pi'(\xi_\omega)U^{-1}\left(U(\xi_\omega)\mathscr{X}_{\{f_\omega=0\}}\right)&=M_{g_\omega}\left(U(\xi_\omega)\mathscr{X}_{\{f_\omega=0\}}\right)=g_\omega U(\xi_\omega)\mathscr{X}_{\{f_\omega=0\}}=0.
\end{align*}
  We apply the unitary $U^{-1}$ on both sides of the previous equation and infer
  \begin{align*}
      \pi'(\xi_\omega)U^{-1}(U(\xi_\omega)\mathscr{X}_{\{f_\omega=0\}})=0\Leftrightarrow U^{-1}(U(\xi_\omega)\mathscr{X}_{\{f_\omega=0\}})\xi_\omega =0
  \end{align*}    
  where in the last equivalence we used Proposition~\ref{prop:pi,pi' for self-adjoint}. Therefore,
    \begin{align*}
       0&=\tau\left(U^{-1}\left(U(\xi_\omega)\mathscr{X}_{\{f_\omega=0\}}\right)\xi_\omega\right)=
       \tau\left(\xi_\omega U^{-1}\left(U(\xi_\omega)\mathscr{X}_{\{f_\omega=0\}}\right)\right) \\ &=\braket{\xi_\omega}{U^{-1}\left(U(\xi_\omega)\mathscr{X}_{\{f_\omega=0\}}\right)}_{L_2(\mathscr{M},\tau)}=\braket{U(\xi_\omega)}{U(\xi_\omega)\mathscr{X}_{\{f_\omega=0\}}}_{L_2(X,\mu)}\\
      &= \int_{X} \overline{U(\xi_\omega)} U(\xi_\omega)\mathscr{X}_{\{f_\omega=0\}}d\mu    =\int_{\{f_\omega=0\}} |U(\xi_\omega)|^2d\mu .
  \end{align*}
   The proof that $\int_{\{f_\phi=0\}} |U(\xi_\phi)|^2d\mu =0$ is almost identical and we do not present it. This finishes the proof of Claim 1. 
   
   Now, we observe that $\int_{X} f_\omega d\nu=1$. Indeed,
\begin{align*} 
  \int_{X}f_\omega d\nu&=\int_{\{f_\omega>0\}}f_\omega d\nu=\int_{\{f_\omega>0\}}f_\omega \frac{|U(\xi_\omega)|^2}{f_\omega}d\mu  \quad \text{ (by Equation~\eqref{eq:measure d nu}})\\
  &=\int_{\{f_\omega>0\}}|U(\xi_\omega)|^2d\mu+\int_{\{f_\omega=0\}} |U(\xi_\omega)|^2d\mu \quad \text{ (by Equation~\eqref{U of x omega zero at f omega zero}})\\
  &=\int_{X}|U(\xi_\omega)|^2d\mu=\braket{U(\xi_\omega)}{U(\xi_\omega)}_{L_2(X,\mu)}\\&=\braket{\xi_\omega}{\xi_\omega}_{L_2(\mathscr{M},\tau)}=\omega(1)=1,
  \end{align*}
where in the last line we used Remark~\ref{existence of hphi,homega}.

Next, we focus on proving that  $\int_{X}f_\phi d\nu=1$. We begin by proving a second claim.\newline
{\bf{\underline{Claim 2:}}}
\begin{align}\label{eq: phi of s M of omega}
\phi(s_\mathscr{M}(\omega))=\int_{\{f_\omega>0\}}f_\phi d\nu.
\end{align}
For the left-hand side we have
\begin{align}\label{eq: equation 1 of prop 3.11}
 \phi(s_{M}(\omega))&=\braket{\xi_\phi}{s_\mathscr{M}(\omega)\xi_\phi} \quad \text{ (by Remark~\ref{existence of hphi,homega}})\notag \\
&=\braket{\xi_\phi}{\pi\left(s_\mathscr{M}(\omega)\right) (\xi_\phi)} \quad  \text{(by Equation~\eqref{Eq: Left Multiplication}) } \notag\\
 &=\|\pi\left(s_\mathscr{M}(\omega)\right)(\xi_\phi)\|_{2}^2 \quad \text{ (since } \pi\left(s_\mathscr{M}(\omega)\right) \text{ is a projection on } L_2\left(\mathscr{M},\tau\right) ) \notag\\
&=\|s_\mathscr{M}(\omega)\xi_\phi\|_{2}^2 \quad  \text{(by Equation~\eqref{Eq: Left Multiplication}). } \end{align}
Then, we  use the polar decomposition of $\overline{S_{\phi,\omega}}$ from Equation~\eqref{eq: polar decomp of S} to derive
\begin{align}\label{eq: equation 2 of prop 3.9}\overline{S_{\phi,\omega}}(\xi_\omega)=J\Delta_{\phi,\omega}^{\frac{1}{2}}(\xi_\omega) \Leftrightarrow s_\mathscr{M}(\omega)\xi_\phi=J\Delta_{\phi,\omega}^{\frac{1}{2}}(\xi_\omega)\end{align}
where the equivalence comes from the formula of $S_{\phi,\omega}$, which was introduced in Definition~\ref{relative modular operator}. We combine Equations~\eqref{eq: equation 1 of prop 3.11}, \eqref{eq: equation 2 of prop 3.9} and the fact that $J$ is isometric to obtain
\begin{align}\label{eq: prop 4.6 first equation of the proof}
\phi(s_{\mathscr{M}}(\omega))&=\|J\Delta_{\phi,\omega}^{\frac{1}{2}}(\xi_\omega)\|_{2}^2=\|\Delta_{\phi,\omega}^{\frac{1}{2}}(\xi_\omega)\|_{2}^2 \notag \\
&=\braket{\Delta_{\phi,\omega}^{\frac{1}{2}}( \xi_{\omega})}{\Delta_{\phi,\omega}^{\frac{1}{2}} (\xi_{\omega})}_{L_2(\mathscr{M},\tau)}
=\braket{U\Delta_{\phi,\omega}^{\frac{1}{2}} (\xi_{\omega})}{U\Delta_{\phi,\omega}^{\frac{1}{2}} (\xi_{\omega})}_{L_2(X,\mu)}.
\end{align}
We apply Corollary~\ref{thm:fDelta} for the square root function and infer  $U\Delta_{\phi,\omega}^{\frac{1}{2}}U^{-1}=M_{\sqrt{f_\phi w(f_\omega)}}.$ Equivalently, $U\Delta_{\phi,\omega}^{\frac{1}{2}}=M_{\sqrt{f_\phi w(f_\omega)}}U$ and by Equation~\eqref{eq: prop 4.6 first equation of the proof} we obtain
\begin{align*}
\phi\left(s_{M}(\omega)\right)&= \braket{M_{\sqrt{f_\phi w(f_\omega)}}U(\xi_\omega) }{M_{\sqrt{f_\phi w(f_\omega)}}U(\xi_\omega)}_{L_2(X,\mu)} \\
&=\int_{X} \overline{\sqrt{f_\phi w(f_\omega)}U(\xi_\omega)}\sqrt{f_\phi w(f_\omega)}U(\xi_\omega)d\mu=\int_{X} f_\phi w(f_\omega) |U(\xi_\omega)|^2d\mu \\
&=\int_{\{f_\omega >0\}} \frac{f_\phi}{f_\omega}|U(\xi_\omega)|^2d\mu \quad \text{(since } w(\lambda)=\frac{1}{\lambda} \text{ for } \lambda>0 \text{ and } w(0)=0 \text{)}  \\
&=\int_{\{f_\omega >0\}} f_\phi d\nu \quad \text{(by Equation~\eqref{eq:measure d nu})}.
\end{align*}
This finishes the proof of Claim 2.\newline
{\bf{\underline{Claim 3:}}}
\begin{align}\label{eq: phi of 1 - s m of omega}
\phi(1-s_\mathscr{M}(\omega))=\int_{\{f_\omega=0\}}f_\phi d\nu.
\end{align}
Let $\xi_\omega=\int_{[0,+\infty)} \lambda dE_{\omega}(\lambda)$ be the spectral decomposition of the vector representative $\xi_\omega$. Then,   $E_\omega\left([0,+\infty)\right)=1\Leftrightarrow E_\omega\left((0,+\infty)\right) +E_\omega(\{0\})=1\Leftrightarrow 1-E_\omega\left((0,+\infty)\right)=E_\omega(\{0\}) $ and by Lemma~\ref{lem: s M of omega}  we know that $E_\omega\left((0,+\infty)\right)=s_\mathscr{M}(\omega)$. Therefore,
\begin{align}\label{eq: prop 4.6 second equation of the proof}
    1-s_{\mathscr{M}}(\omega)=E_\omega(\{0\}).
\end{align}
We now focus on the term $\phi\left(1-s_\mathscr{M}(\omega)\right)$.
\begin{align}\label{eq: prop 4.6 third equation of the proof}
    \phi\left(1-s_\mathscr{M}(\omega)\right)&=\phi\left( E_\omega(\{0\})\right)=\tau\left(\xi_\phi^2   E_\omega(\{0\})\right) \quad  \text{(by Equations~\eqref{eq: prop 4.6 second equation of the proof} and~\eqref{eq: equation for vector representative})} \notag\\
    &=\braket{\xi_\phi}{\xi_\phi E_\omega(\{0\})}_{L_2(\mathscr{M},\tau)}=\braket{\xi_\phi}{ \pi' (E_\omega(\{0\}))( \xi_\phi) }_{L_2(\mathscr{M},\tau)}
\end{align}
where in the last equality we used Equation~\eqref{Eq: right multiplication} (since $\xi_\phi\in L_2(\mathscr{M},\tau)$ and   $E_\omega (\{0\})\in \mathscr{M}$).
We observe that $E_\omega(\{0\})=\int_{\{0\}} 1 dE_\omega(\lambda)= \mathscr{X}_{\{0\}} (\xi_\omega) $ and  apply the right multiplication $\pi'$ on both sides  to obtain 
\begin{align}\label{eq: prop 4.6 fourth equation of the proof}
    \pi'\left(E_\omega(\{0\})\right)&=\pi'\left(\mathscr{X}_{\{0\}} (\xi_\omega)\right)=\mathscr{X}_{\{0\}} \left(\pi'(\xi_\omega)\right)\text{ (by Equation~\eqref{Eq: pi of f of x})} \notag\\
    &=\mathscr{X}_{\{0\}} \left(U^{-1}M_{g_\omega}U\right) \quad \text{(since } U\pi'(\xi_\omega)U^{-1}=M_{g_\omega} \text{ by Equation~\eqref{Eq: multiplication form we use})} \notag\\ &= U^{-1}\mathscr{X}_{\{0\}}\left(M_{g_\omega}\right)U \quad \text{ (by Lemma~\ref{lem: applying the Borel functional calculus})} \notag\\
    &=U^{-1}M_{\mathscr{X}_{\{0\}} \left(g_\omega\right)}U   \quad \text{(by \cite[Example 5.3]{Schmudgen})} \notag\\
    &=U^{-1}M_{\mathscr{X}_{\{g_\omega=0\}}}U.
\end{align}

We combine Equations~\eqref{eq: prop 4.6 third equation of the proof}, ~\eqref{eq: prop 4.6 fourth equation of the proof} and infer \begin{align*}
     \phi(1-s_\mathscr{M}(\omega))&=\braket{\xi_\phi}{  U^{-1}M_{\mathscr{X}_{\{g_\omega=0\}}}U (\xi_\phi) }_{L_2(\mathscr{M},\tau)}=\braket{U(\xi_\phi)}{  M_{\mathscr{X}_{\{g_\omega=0\}}}U (\xi_\phi) }_{L_2(X,\mu)}\\
     &=\int_{X} \overline{U(\xi_\phi)} \mathscr{X}_{\{g_\omega=0\}} U(\xi_\phi) d\mu\\
     &=\int_{\{f_\omega=0\}} |U(\xi_\phi)|^2d\mu \quad \text{ (since } \{g_\omega=0\} =\{f_\omega=0\} ) \\
     &= \int_{\{f_\phi>0,f_\omega=0\}} |U(\xi_\phi)|^2d\mu+ \int_{\{f_\phi=0,f_\omega=0\}} |U(\xi_\phi)|^2d\mu\\
     &= \int_{\{f_\phi>0,f_\omega=0\}} |U(\xi_\phi)|^2d\mu \quad \text{ (by Claim 1)} \\
     &= \int_{\{f_\phi>0,f_\omega=0\}} f_\phi \frac{|U(\xi_\phi)|^2}{f_\phi }d\mu =\int_{\{f_\phi >0,f_\omega =0\}} f_{\phi}d\nu \quad \text{(by Equation~\eqref{eq:measure d nu})}\\
     &= \int_{\{f_\omega =0\}} f_{\phi}d\nu.
     \end{align*}
which concludes the proof of Claim 3.

We combine Equations~\eqref{eq: phi of s M of omega}, \eqref{eq: phi of 1 - s m of omega} to finish the proof of the proposition.
\begin{align*}
    \int_{X}f_\phi d\nu&=\int_{\{f_{\omega}>0\}}f_\phi d\nu+ \int_{\{f_\omega =0\}} f_{\phi}d\nu =\phi\left(s_\mathscr{M}(\omega)\right)+\phi\left(1-s_\mathscr{M}(\omega)\right)\\
    &=\phi\left(s_\mathscr{M}(\omega)+1-s_\mathscr{M}(\omega)\right)=\phi(1)=1.
\end{align*}
\end{proof}
We are now ready to  prove our main theorem, Theorem~\ref{maintheorem}.
\begin{proof}[Proof of Theorem~\ref{maintheorem}] 
Let $(X,\Sigma,\mu)$ be the measure space that we obtained in Theorem~\ref{prop: application of the multiplication form of spectral theorem} and $\nu$ be the measure we defined in Equation~\eqref{eq:measure d nu}. First, we prove that $(X,\Sigma,\nu)$ is a $\sigma$-finite measure space (which is required  in the definition of the classical $f$-divergence). It suffices to show that there exists a function $g$ on $X$ such that $g(x)>0$ for all $x\in X$ and $\int_{X} g d\nu<+\infty$. Let $f_\phi$, $f_\omega$ be the functions we obtained in Corollary~\ref{thm:fDelta}. We define the function $g:X\to \mathbb{R}$ by the formula
\[g(x):= \begin{cases} 
f_{\omega}(x) & \text{if }  x\in\{f_{\omega}> 0\}\\
f_{\phi}(x) & \text{if } x\in \{f_{\phi}> 0,f_{\omega}=0\}\\
\quad 1 & \text{if } x\in \{f_{\phi}=f_{\omega}=0\}\ .
\end{cases} \]
Obviously, $g(x)>0$ for all $x\in X$ and 
\begin{align*}
    \int_{X}gd\nu&=\int_{\{f_{\omega}> 0\}}f_\omega d\nu +\int_{\{f_{\phi}> 0,f_\omega=0\}}f_\phi d\nu+\int_{\{f_{\phi}= 0,f_\omega=0\}}1 d\nu\\
    &=\int_{X}f_\omega d\nu +\int_{\{f_{\phi}> 0,f_\omega=0\}}f_\phi d\nu+0 \quad \text{(by Equation~\eqref{eq:measure d nu})}\\ 
    &\leq \int_{X}f_\omega d\nu +\int_{X}f_\phi d\nu=1+1=2<+\infty  \quad \text{(by Proposition~\ref{prop: f phi,omega are distributions}).}
\end{align*}
Therefore, $(X,\Sigma,\nu)$ is a $\sigma$-finite measure space.

Let $\xi_\phi$, $\xi_\omega$ be the vector representatives of $\phi$, $\omega$ respectively and $\Delta_{\phi,\omega}$ the relative modular operator with respect to $\phi$ and $\omega$. We recall that by Equation~\eqref{eq: formula of the quantum f divergence},   the quantum $f$-divergence between $\phi$ and $\omega$ is defined as
\begin{align*}
S_f(\phi\|\omega):=\langle\langle\xi_\omega|f(\Delta_{\phi,\omega})(\xi_\omega) \rangle\rangle+f(0^{+})\omega(1-s_{\mathscr{M}}(\phi))+ f^{'}(+\infty)\phi(1-s_{\mathscr{M}}(\omega))
\end{align*}
following the notation in~\eqref{eq:convex-f-notations-1}. In  Theorem~\ref{lem: first term}, we proved that 
\begin{align}\label{eq: first term finally}
   \langle\langle\xi_\omega|f(\Delta_{\phi,\omega})(\xi_\omega) \rangle\rangle&=\int_{\{f_\phi,f_\omega>0\}} f\left(\frac{f_\phi}{f_\omega}\right) |U(\xi_\omega)|^2d\mu \notag \\
    &=\int_{\{f_\phi,f_\omega>0\}} f_\omega f\left(\frac{f_\phi}{f_\omega}\right) \frac{|U(\xi_\omega)|^2}{f_\omega}d\mu \notag \\
    &=\int_{\{f_\phi,f_\omega>0\}} f_\omega f\left(\frac{f_\phi}{f_\omega}\right)d\nu.
\end{align} 

We claim that 
\begin{align}\label{eq: 1.3 first equation of the proof}
\omega\left(1-s_{\mathscr{M}}(\phi)\right)=\int_{\{f_\phi=0\}} f_\omega d\nu.
\end{align}
In the proof of Proposition~\ref{prop: f phi,omega are distributions}, we showed that $\phi\left(1-s_{\mathscr{M}}(\omega)\right)=\int_{\{f_\omega=0\}} f_\phi d\nu$ (Claim 3 of the proof). The proof of Equation~\eqref{eq: 1.3 first equation of the proof} is almost identical  and we do not present it. One just needs to repeat the exact same steps as in the proof of Claim 3 of the previous proposition but swap $\phi$, $\omega$ and use the left multiplication $\pi$ instead of the right multiplication $\pi'$ (wherever $\pi'$ is used). This finishes the proof of Equation~\eqref{eq: 1.3 first equation of the proof}.

In the next steps, we finish the proof of the theorem. 
 \begin{align*}
S_f(\phi\|\omega)&=\langle\langle\xi_\omega|f(\Delta_{\phi,\omega})(\xi_\omega) \rangle\rangle+f(0^{+})\omega\left(1-s_{\mathscr{M}}(\phi)\right)+ f^{'}(+\infty)\phi\left(1-s_{\mathscr{M}}(\omega)\right) \\
  &= \int_{\{f_\phi,f_\omega>0\}} f_\omega f\left(\frac{f_\phi}{f_\omega}\right)d\nu+ f(0^{+})\int_{\{f_\phi =0\}} f_{\omega}d\nu+ f^{'}(+\infty) \int_{\{f_\omega =0\}} f_{\phi}d\nu\\
  &\text{ (by Equations~\eqref{eq: first term finally}, ~\eqref{eq: 1.3 first equation of the proof} and ~\eqref{eq: phi of 1 - s m of omega})}\\
  &=\int_{\{f_\phi,f_\omega>0\}} f_\omega f\left(\frac{f_\phi}{f_\omega}\right)d\nu+ f(0^{+})\int_{\{f_\phi =0,f_\omega >0\}} f_{\omega}d\nu+ f^{'}(+\infty) \int_{\{f_\phi >0, f_\omega =0\}} f_{\phi}d\nu\\
  &=  D_f(P_\phi \|P_\omega ) \quad \text{ (by Equation~\eqref{eq: formula of the classical divergence}).}
\end{align*}
 Therefore, the quantum  \texorpdfstring{\ensuremath{f}}{}-divergence between $\phi$ and $\omega$ is equal to the classical  \texorpdfstring{\ensuremath{f}}{}-divergence between $P_\phi$ and $P_\omega$, which are classical states by Proposition~\ref{prop: f phi,omega are distributions}.
\end{proof} 

\begin{remark}\label{rmk: no uniqueness of NS distributions}
    The classical states $P_\phi $ and $P_\omega $ 
that appear in the proof of Theorem~\ref{maintheorem} 
are called  {\bf{Nussbaum-Szko{\l}a distributions}} associated with the states $\phi$ and $\omega$. We would like to mention that the distributions we obtained are by no means unique. 

Indeed, in Theorem 3.7, we proved the existence of a measure space $(X, \Sigma,\mu)$, a unitary map  $U:L_2\left(\mathscr{M},\tau\right)\to L_2\left(X,\mu \right)$ and two Borel measurable functions $g_\phi$, $g_\omega: X\to [0,+\infty)$ that satisfy Equations~\eqref{Eq: multiplication form we use}, \eqref{eq: square root of Delta mult oper} where $\phi$ and $\omega$ are normal states on a semifinite von~Neumann algebra $\mathscr{M}$ with a faithful semifinite normal trace $\tau$.  It is clear that one could possibly construct different 
\begin{align*}
   &\text{ measure space  } (X',\Sigma',\mu'),\\
   &\text{ unitary map } V:L_2(\mathscr{M},\tau)\to L_2(X',\mu') \text{ and }\\
    &\text{ functions } g_\phi', g_\omega':X'\to [0,+\infty)
\end{align*}
  such that Equations~\eqref{Eq: multiplication form we use}, \eqref{eq: square root of Delta mult oper} hold. Then, the reasoning performed in Sections 3 and 4 for the functions $g_\phi$ and $g_\omega$ can be repeated word for word for the functions $g_\phi'$ and $g_\omega'$. That would create some other Nussbaum-Szko{\l}a distributions associated with the states $\phi$, $\omega$.

Nevertheless, there is some sort of uniqueness in our construction. The multiplication operators $M_{g_\phi}$ and $M_{g_\omega}$ that we obtained in Theorem~\ref{prop: application of the multiplication form of spectral theorem} are unique up to a unitary map. Indeed, following the notation of the previous paragraphs, observe that 
\begin{align*}
    VU^{-1}: L_2(X,\mu)\to L_2(X',\mu')
\end{align*}
is a unitary map and we prove that $M_{g_\phi'}=VU^{-1}M_{g_\phi}\left(VU^{-1}\right)^{-1}$. Indeed, 
\begin{align*}
VU^{-1}M_{g_\phi}\left(VU^{-1}\right)^{-1}&=VU^{-1}M_{g_\phi}UV^{-1}\\
&=V \pi\left(\xi_\phi\right) V^{-1} \quad \text{ (by Equation~\eqref{Eq: multiplication form we use})}\\
&=M_{g_\phi'} \quad \text{ (by Equation~\eqref{Eq: multiplication form we use} for the unitary } V).
\end{align*} 
Using exactly the same reasoning and the second part of Equation~\eqref{Eq: multiplication form we use}, one can easily prove that $M_{g_\omega'}=VU^{-1}M_{g_\omega}\left(VU^{-1}\right)^{-1}$. Thus, the multiplication operators $M_{g_\phi}$ and $M_{g_\omega}$ that were obtained in Theorem~\ref{prop: application of the multiplication form of spectral theorem} and which are the key to our construction, are unique up to a unitary map. 
\end{remark}

\section{Applications}
The usefulness of  Nussbaum-Szko{\l}a distributions  is providing a tool for directly obtaining quantum versions of several results which are already known for classical $f$-divergences. Our main Theorem~\ref{maintheorem}, extends the concept of Nussbaum-Szko{\l}a distributions for normal states on any semifinite von~Neumann algebra, while it was previously introduced only for normal states on the von~Neumann algebra $\mathbb{B}(\mathscr{H})$. 
We provide a few applications by demonstrating how one can use our Theorem~\ref{maintheorem} to derive inequalities for quantum $f$-divergences between states on a semifinite von~Neumann algebra starting from classical $f$-divergence inequalities  between classical probability distributions. 

Thoughout this section, we use the letters $\phi$,  $\omega$ to denote  normal quantum states on a semifinite von~Neumann algebra and the letters $P,\text{ }Q$ for classical probability distributions. Let $P_\phi$, $P_\omega$ be  Nussbaum-Szko{\l}a distributions associated with  the states $\phi$, $\omega$.  

First, we review some important examples of quantum $f$-divergences. 
 \begin{itemize}
     \item By \cite[Equation (1.2), also end of Remark 2.7]{Hiai-2018} (see also \cite[Equation (41)]{sason-verdu-2016} and \cite[Example 3]{Liese-Vajda-2006} for the classical case), the relative entropy $D(\phi \| \omega)$ is the quantum $f$-divergence $S_{f}(\phi\|\omega)$ where $f(t)=t\log(t)$ for $t\in (0,+\infty)$.

     We would like to mention that there are articles in the bibliography (for example \cite{araki1977relative}) which use the convex function $g(t)=-\log(t)$, for $t\in (0,+\infty)$, to define the relative entropy. This can be done  by reversing the order of the states $\phi$ and $\omega$ in the definition of the quantum $f$-divergence. In other words, 
     \begin{align}\label{eq: reversing the order}
         S_{t\log(t)} \left(\phi\| \omega \right)= S_{-log(t)} \left(\omega \| \phi\right).
     \end{align}
     Indeed, if $f(t)=t\log(t)$, for $t\in (0,+\infty)$, then we observe that 
     \begin{align*}\widetilde{f}(t):=tf\left(\frac{1}{t}\right)=-\log(t), \quad \text{ for } t\in (0,+\infty)
     \end{align*}
     and by \cite[Proposition 2.4]{Hiai-2018} we have $S_{f}(\phi\|\omega)=S_{\widetilde{f}}(\omega\|\phi)$. This proves Equation~\eqref{eq: reversing the order}.
     \item The $\chi^2$-divergence $\chi^2(\phi \| \omega)$ is the quantum $f$-divergence $S_{f}(\phi\|\omega)$ where $f(t)=(t-1)^2$.
     \item The total variation distance $|\phi-\omega|$ is the quantum $f$-divergence $S_{f}(\phi\|\omega)$ where $f(t)=|t-1|$.
 \end{itemize}

 By \cite[Inequality (5)]{sason-verdu-2016} we know that 
 \[ D(P\|Q)\leq \log\left(1+\chi^2(P\|Q)\right).\]
Therefore, 
\begin{align}\label{eq: 1st application equation}
    D(\phi\|\omega)&=  D(P_\phi \|P_\omega ) \quad \text{ (by Theorem~\ref{maintheorem})} \notag\\
                   &\leq \log\left(1+\chi^2(P_\phi \|P_\omega)\right) \notag\\
                   &= \log\left(1+\chi^2(\phi \|\omega)\right) \quad \text{ (by Theorem~\ref{maintheorem})}.
\end{align}

By \cite[Inequality (13)]{sason-verdu-2016} we know that 
 \[ D(P\|Q)\leq \frac{1}{2}\left(|P-Q|+\chi^2(P\|Q)\right)log(e).\]
Therefore, 
\begin{align}\label{eq: 2nd application equation}
    D(\phi\|\omega)&=D(P_\phi\|P_\omega )  \quad \text{ (by Theorem~\ref{maintheorem})} \notag\\
                   &\leq \frac{1}{2}\left(|P_\phi-P_\omega|+\chi^2(P_\phi\|P_\omega)\right)log(e) \notag \\
                    &= \frac{1}{2}\left(|\phi-\omega|+\chi^2(\phi\|\omega)\right)log(e) \quad \text{ (by Theorem~\ref{maintheorem})}.
\end{align}

In addition to Inequalities~\eqref{eq: 1st application equation} and \eqref{eq: 2nd application equation}, one can derive many more such inequalities. For example, one can find plenty of classical $f$-divergence inequalities in \cite{sason-verdu-2016} and use our Theorem~\ref{maintheorem} to generalize the inequalities to the quantum case for normal states on any semifinite von~Neumann algebra. 

Furthermore, one could also prove existing properties of the quantum $f$-divergence by applying Theorem~\ref{maintheorem} 
to their classical counterparts. For example, now we give a very simple proof of \cite[Proposition 2.4]{Hiai-2018} in the case of semifinite von~Neumann algebras. According to this proposition, $ S_f (\phi\|\omega)=S_{\widetilde{f}} (\omega\|\phi)$ for every convex function $f:(0,+\infty)\to \mathbb{R}$ where $\widetilde{f}(t):=tf\left(\frac{1}{t}\right)$, for $t\in (0,+\infty)$. Indeed, by \cite[Theorem 4]{Liese-Vajda-2006} we know that 
 \[ D_f (P\|Q)=D_{\widetilde{f}} (Q\|P).\]
Therefore, 
\begin{align*}
   S_f (\phi\|\omega)&=  D_f(P_\phi \|P_\omega ) \quad \text{ (by Theorem~\ref{maintheorem})}\\
                   &=D_{\widetilde{f}} (P_{\omega}\|P_{\phi})\\
                  &= S_{\widetilde{f}} (\omega\|\phi)  \quad \text{ (by Theorem~\ref{maintheorem})}.
\end{align*}

\section{Closing Remarks}
In this article, we prove that the quantum \texorpdfstring{\ensuremath{f}}{}-divergence between two normal states on a semifinite von~Neumann algebra is equal to the classical \texorpdfstring{\ensuremath{f}}{}-divergence between two corresponding classical states. The divergences of states on these algebras have not been explored in depth, but there is some physical motivation to study them.  For example, the  Type~$II_1$  factor appears in super JT-gravity \cite{Penington2024}. Furthermore, Type~$II$ algebras can appear in  some random matrix models \cite{BibEntry2005Feb}, in quantum field theory \cite{Longo2023} and in black hole physics \cite{Witten2022Oct}.

The main open problem that is related to our work is whether Theorem~\ref{maintheorem} still holds if we drop the assumption that the von~Neumann algebra is semifinite. Does Theorem~\ref{maintheorem} hold for normal states $\phi$, $\omega$ on a general von~Neumann algebra? It would not be surprising if the answer to this question was affirmative. The main obstacle that has not allowed us to prove the main theorem for any von~Neumann algebra is the lack of a convenient formula for the relative modular operator in the general case.

\subsection*{Acknowledgment}
The authors thank the anonymous referee for the valuable comments and suggestions that significantly improved this article.


\end{document}